\documentclass[a4paper,USenglish]{lipics-v2016}
 
\usepackage{microtype}
\usepackage{colonequals}
\usepackage{bm}
\usepackage{tikz}
\usepackage{apxproof}
\usepackage{xcolor}
\usepackage{amssymb}
\usepackage{tikz-qtree}

\usetikzlibrary{positioning, fit, calc, shapes, arrows, arrows.meta}
\title{Enumeration on Trees under Relabelings\footnote{This
is the complete version with proofs of the corresponding ICDT'18
publication~\cite{amarilli2018enumeration}.}}
\author[1]{Antoine Amarilli}
\author[2]{Pierre Bourhis}
\author[3]{Stefan Mengel}
\affil[1]{LTCI, Télécom ParisTech, Université Paris-Saclay; Paris, France}
\affil[2]{CRIStAL, CNRS UMR 9189 \& Inria Lille; Lille, France}
\affil[3]{CNRS, CRIL UMR 8188; Lens, France}

\EventEditors{Benny Kimelfeld and Yael Amsterdamer}
\EventNoEds{2}
\EventLongTitle{21st International Conference on Database Theory (ICDT 2018)}
\EventShortTitle{ICDT 2018}
\EventAcronym{ICDT}
\EventYear{2018}
\EventDate{March 26-–29, 2018}
\EventLocation{Vienna, Austria}
\EventLogo{}
\SeriesVolume{98}
\ArticleNo{5}

\hypersetup{
    colorlinks,
    linkcolor={red!50!black},
    citecolor={blue!50!black},
    urlcolor={blue!30!black}
}

\Copyright{Antoine Amarilli, Pierre Bourhis, Stefan Mengel}

\usepackage{chngcntr}

\theoremstyle{theorem}
\newtheoremrep{theorem}{Theorem}
\newtheoremrep{lemma}[theorem]{Lemma}
\newtheoremrep{proposition}[theorem]{Proposition}
\newtheoremrep{claim}[theorem]{Claim}
\newtheoremrep{remark}[theorem]{Remark}
\newtheoremrep{observation}[theorem]{Observation}
\newtheoremrep{corollary}[theorem]{Corollary}

\counterwithin{theorem}{section}

\makeatletter
\newcommand*{\defeq}{\mathrel{\rlap{%
  \raisebox{0.3ex}{$\m@th\cdot$}}%
  \raisebox{-0.3ex}{$\m@th\cdot$}}%
  =}
\makeatother

\newcommand{\card}[1]{\left|{#1}\right|}

\newcommand{\NN}{\mathbb{N}}

\renewcommand{\phi}{\varphi}

\newcommand{\calA}{\mathcal{A}}

\newcommand{\h}{\mathrm{h}}

\newcommand{\first}{\mathsf{first}}
\newcommand{\last}{\mathsf{last}}
\newcommand{\pnext}{\mathsf{next}}

\newcommand{\var}{\mathrm{var}}
\newcommand{\bvar}{\mathrm{bvar}}
\newcommand{\avar}{\mathrm{svar}}

\newcommand{\reach}{\mathrm{reach}}

\renewcommand{\r}{\mathrm{r}}

\newcommand{\anc}{\mathcal{A}}

\newcommand{\enu}{\mathsf{\vphantom{pf}enu}}
\newcommand{\fix}{\mathsf{\vphantom{pf}fix}}
\newcommand{\upd}{\mathsf{\vphantom{pf}upd}}

\newcommand{\inp}{\mathrm{inp}}

\newcommand{\aplus}{\cup}
\newcommand{\atimes}{\times}

\newcommand{\snull}{\mathsf{null}}

\newcommand{\relprod}{\times_{\mathrm{rel}}}

\newcommand{\opath}{\rightarrow^*_\ctimes}

\newcommand{\Sigmafix}{\Sigma_{\mathsf{e,u,f}}}

\newcommand{\node}{\mathsf{node}}
\newcommand{\nfix}{\mathsf{nfix}}
\newcommand{\dom}{\mathsf{dom}}

\newcommand{\gset}[1]{\mathrm{S}(#1)}
\newcommand{\gsetv}[2]{\mathrm{S}_{#1}(#2)}

\newcommand{\bval}[2]{\mathrm{V}_{#1}(#2)}

\newcommand{\sgb}[1]{$\langle B\!\!:\!\!#1\rangle$}
\newcommand{\sgx}[1]{$\langle x\!\!:\!#1\rangle$}

\newcommand{\splus}{\oplus}
\newcommand{\stimes}{\otimes}

\newcommand{\ctimes}{\boxtimes}

\begin{document}

\maketitle
\begin{abstract}
  We study how to evaluate MSO queries with free variables on trees, within the
framework of enumeration algorithms. Previous work has shown how to enumerate
answers with linear-time preprocessing and delay linear in the size of each
output, i.e., constant-delay for free first-order variables. We extend this
result to support \emph{relabelings}, a restricted kind of update operations on
trees which allows us to change the node labels. Our main result shows that we
can enumerate the answers of MSO queries on trees with linear-time preprocessing
and delay linear in each answer, while supporting node relabelings in logarithmic time. To
prove this, we reuse the circuit-based enumeration structure from our earlier
work, and develop techniques to maintain its index under node relabelings. We
also show how enumeration under relabelings can be applied to evaluate practical
query languages, such as aggregate, group-by, and parameterized queries.

\end{abstract}

\section{Introduction}
\label{sec:introduction}
Enumeration algorithms are 
a common way to compute large query results on databases,
see, e.g.,~\cite{Segoufin14}. Instead of computing all
results, these algorithms 
compute results one after the other, while ensuring
that the time between two
successive results
(the \emph{delay}) remains small. Ideally, the delay
should be \emph{linear} in the size of each produced solution, and 
independent
of the size of the input database.
To make this possible, 
enumeration algorithms can
build an index structure on the database during a
\emph{preprocessing 
phase} that 
ideally runs in linear time.

Most enumeration algorithms assume that the input database
will not change.
If we \emph{update} the database,
we must re-run the
preprocessing phase from scratch,
which is unreasonable in practice.
Losemann and Martens~\cite{losemann2014mso}
proposed the first enumeration algorithm that 
addresses this issue:
they study monadic second-order (MSO) query evaluation on trees,
and show that the
index structure for enumeration can be maintained under updates. More precisely,
they can update the index in time polylogarithmic in 
the input tree~$T$ (much better than re-running the linear preprocessing).
The tradeoff is that their delay is also polylogarithmic in~$T$,
whereas the delay can be independent of~$T$ when there are no
updates~\cite{bagan2006mso}.

This result of~\cite{losemann2014mso} leads to a natural question: does the support
for updates inherently increase the delay of enumeration algorithms?
This is not always the case: e.g., when evaluating
first-order queries (plus modulo-counting quantifiers) on bounded-degree
databases, updates can be applied in constant time~\cite{berkholz2017answering2} 
and the delay is constant, as in the case without updates~\cite{durand2007first,kazana2011first}.
However, when evaluating conjunctive queries (CQs) on arbitrary
databases, supporting updates has a cost:
under complexity-theoretic assumptions,
the class of CQs with efficient enumeration
under updates~\cite{berkholz2017answering}
is a strict subclass 
of the class of CQs for the case without updates~\cite{bagan2007acyclic}.
Could the same be true of MSO on trees, as \cite{losemann2014mso} would
suggest?

In this work, we answer this question in the negative, for a restricted update
language. Specifically, 
we show an enumeration algorithm for MSO on trees with the same delay
as in the case without updates~\cite{bagan2006mso}, while supporting updates with a better complexity
than~\cite{losemann2014mso} (see detailed comparison of results in Section~\ref{sec:problem}).
The tradeoff is that we only allow updates that change the labels of
nodes, called \emph{relabelings},
unlike~\cite{losemann2014mso} where updates can also insert and delete leaves.
We still show how these relabelings are useful to evaluate
practical query languages, such as
\emph{parameterized} queries and \emph{group-by queries with aggregates}.
A \emph{parameterized} query allows the user to specify some parameters for
the evaluation (e.g., select some positions on the tree). 
Our results support such queries:
we can model the parameters as labels and apply relabeling updates when the
user changes the parameters. 
A \emph{group-by query with aggregates} 
partitions the set of results into groups based on an attribute,
and 
computes some aggregate quantity on each group (e.g., a sum). 
We show how to 
enumerate the results
of such queries.
For groups, our techniques can handle them with one single enumeration structure
using relabelings to switch groups.
For aggregates, we can efficiently compute and maintain them
in arbitrary semirings;
this problem was left open by~\cite{losemann2014mso} even for counting,
and is practically relevant in its own right~\cite{nikoli2017incremental}. Of
course, by Courcelle's theorem~\cite{courcelle1990monadic1}, our
results generalize to MSO queries on bounded-treewidth data
(see~\cite{amarilli2017circuit_extended}), where relabelings mean adding or
removing unary facts (i.e., the tree decomposition is unchanged).

The proof of our main result follows the approach of~\cite{amarilli2017circuit} 
and is inspired by knowledge compilation in artificial intelligence
and by factorized representations in database theory. Specifically, we
encode knowledge (in our case, the query result) as a circuit in a restricted class,
and we then use the circuit for efficient reasoning and for aggregates as in~\cite{deutch2014circuits}.
In~\cite{amarilli2017circuit}, 
we have used this circuit-based approach to recapture existing enumeration results for MSO on
trees \cite{bagan2006mso,kazana2013enumeration}.
In this work, we refine the approach and show
that it can support
updates.
Our key new ingredient are \emph{hybrid circuits}:
they have both \emph{set-valued}
gates that represent the values to enumerate, and
\emph{Boolean} gates that encode the tree labels which can be updated.
We first show that we can efficiently compute such circuits to capture the possible
results of an MSO query under all possible labelings of a tree.
Second, we show how to efficiently enumerate the set of assignments captured by
these circuits, also supporting updates that toggle the Boolean gates
affected by a relabeling. We also introduce some standalone tools, e.g., 
a lemma to \emph{balance} the input trees to MSO queries (Lemma~\ref{lem:balancing}),
ensuring that hybrid circuits have logarithmic depth so that changes can be
propagated quickly; and a constant-delay enumeration algorithm for
reachability in forests under updates
(Section~\ref{sec:reachability}).

\subparagraph*{Paper structure.}
We start with preliminaries in Section~\ref{sec:prelim}, and define our
problem and give our main result in Section~\ref{sec:problem}.
In Section~\ref{sec:provenance}, we review the set-valued provenance circuits
of~\cite{amarilli2017circuit}, and show our balancing lemma.
We introduce hybrid circuits in
Section~\ref{sec:hybrid}, and
show in
Section~\ref{sec:enumeration}
how to use them for enumeration under updates,
 using a standalone reachability indexing scheme on forests
given in Section~\ref{sec:reachability}. Having shown our main result,
we outline its consequences for application-oriented query languages in
Section~\ref{sec:applications} and conclude in Section~\ref{sec:conclusion}.

\section{Preliminaries}
\label{sec:prelim}
\subparagraph*{Trees, queries, answers, assignments.}
In this work, unless otherwise specified, a \emph{tree} is always 
binary, rooted, ordered, and full.
Let $\Gamma$ be a finite set called a \emph{tree alphabet}.
A \emph{$\Gamma$-tree} $(T, \lambda)$ is a pair of a tree $T$
and of a \emph{labeling function} $\lambda$ that maps each
node $n$ of~$T$ to a \emph{set of labels} $\lambda(n) \subseteq \Gamma$.
We often abuse notation
and identify $T$ to its node set, e.g., write $\lambda$ as a function
from~$T$ to the powerset~$2^{\Gamma}$ of~$\Gamma$; we may
also omit~$\lambda$ and write
the $\Gamma$-tree as just~$T$.

We consider queries in \emph{monadic second-order
logic} (MSO)
on the \emph{signature} of $\Gamma$-trees: it features two binary relations $E_1$
and $E_2$ denoting the first and second child of each internal node, and a
unary relation $P_l$ for each $l \in \Gamma$ denoting the nodes 
that carry label~$l$ (i.e., nodes~$n$ for which $l \in \lambda(n)$).
MSO extends \emph{first-order
logic}, which builds formulas from atoms of this signature and from equality atoms,
using the Boolean connectives and existential and universal quantification over
nodes. Formulas in MSO can also use second-order quantification over sets of nodes,
written as second-order variables.
For instance, on $\Gamma = \{l_1, l_2, l_3\}$, we can express in MSO
that every node carrying labels $l_1$ and $l_2$
has a descendant carrying label~$l_3$.

In this work, we 
study MSO \emph{queries}, i.e., MSO formulas with free variables. The free variables can be first-order or
second-order, but we can
rewrite any MSO query $Q(\mathbf{x}, \mathbf{Y})$ to ensure that all free variables are
second-order: for instance as $Q'(\mathbf{X}, \mathbf{Y}):
\exists \mathbf{x} ~ 
\bigwedge_i \textrm{Sing}(X_i, x_i)
\land Q(\mathbf{x}, \mathbf{Y})$, where $\textrm{Sing}(X, x)$ asserts
that $X$ is exactly the singleton set $\{x\}$.
Hence, we usually assume without loss of generality that MSO
queries only have second-order free variables.

Given a $\Gamma$-tree $T$ and an MSO query $Q(X_1, \ldots, X_m)$, an
$m$-tuple $\mathbf{B} = B_1, \ldots, B_m$ of subsets of~$T$ is an \emph{answer} of~$Q$ on~$T$,
written $T \models Q(\mathbf{B})$, if $T$ satisfies $Q(\mathbf{B})$ in the usual
logical sense. It will be more convenient to represent each answer as an
\emph{assignment}, which is a set of pairs called \emph{singletons} that
indicate that an element is in the interpretation of a variable. Formally, given
an $m$-tuple $\mathbf{B}$ of subsets of~$T$, the corresponding assignment is
$\{\langle X_i: n \rangle
\mid 1 \leq i \leq m \text{~and~} n \in B_i\}$. We can convert each assignment in linear
time to the corresponding answer and vice-versa, so we will use the
assignment representation throughout this work.
Our goal is to compute the set of assignments of~$Q$ on~$T$, which we call
the \emph{output} of~$Q$ on~$T$;
we abuse notation and write it $Q(T)$.
We measure the complexity of this task in \emph{data
complexity}, i.e., as a function of the input tree~$T$, with the query~$Q$ being
fixed.

\subparagraph*{Enumeration.}
The output of an MSO query can be huge, so we work in the setting of \emph{enumeration
algorithms} \cite{Wasa16,Segoufin14}
which we present following~\cite{amarilli2017circuit}.
As usual for enumeration algorithms~\cite{Segoufin14},
we work
in the RAM model with uniform cost measure (see,
e.g.,~\cite{AhoHU74}), 
where pointers, numbers, labels for elements and facts, etc.,
have constant size.

An \emph{enumeration algorithm with linear-time preprocessing} for a fixed MSO query
$Q(\mathbf{X})$ on~$\Gamma$-trees takes as input a $\Gamma$-tree $T$ and computes the output
$Q(T)$ of~$Q$ on~$T$.
It consists of two
phases.
First, the \emph{preprocessing phase} 
takes $T$ as input
and produces in \emph{linear time} a
data structure~$J$ called the \emph{index}, and an initial \emph{state}
$s$. 
Second, 
the \emph{enumeration phase}
repeatedly calls an algorithm $\calA$. Each call to~$\calA$ takes as input the
index $J$ and the current state $s$,
and returns one assignment and a new state $s'$: a special state value indicates that
the enumeration is over so $\calA$ should not be called again.
The assignments produced by the successive calls to~$\calA$
must be exactly the elements of $Q(T)$, with no
duplicates.

We say that the enumeration
algorithm has \emph{linear delay} if the
time to produce each new assignment $A$ is linear in its cardinality
$\card{A}$,
and is independent of~$T$.
In particular, if all answers to~$Q$ are tuples of singleton sets (for instance, if
$Q$ is the translation of a MSO query where all free variables are first-order),
then the cardinality of each assignment is constant (it is the arity of~$Q$).
In this case,
the enumeration algorithm 
must produce each assignment with constant delay: this is called \emph{constant-delay
enumeration}.
The \emph{memory usage} of an enumeration algorithm is the maximum number of
memory cells used during the 
enumeration phase (not counting the index $J$, which resides in
read-only memory), expressed as a function of the size of the largest assignment
(as in~\cite{bagan2006mso}): we say that the enumeration algorithm has
\emph{linear memory} if its memory usage is linear in the size of the largest
assignment.

Previous works have studied enumeration for MSO on
trees.
Bagan~\cite{bagan2006mso} showed that for any fixed MSO query $Q(\mathbf{X})$,
given a $\Gamma$-tree $T$,
we can enumerate the output of~$Q$ on~$T$
with linear delay and memory, i.e.,
constant delay and memory when all free variables are first-order.
This result was re-proven by Kazana and Segoufin~\cite{kazana2013enumeration}
via a result of Colcombet~\cite{colcombet2007combinatorial}, and a third proof
via provenance circuits was recently proposed by the present
authors~\cite{amarilli2017circuit}.

\section{Problem Statement and Main Result}
\label{sec:problem}
Our goal is to address a limitation of these existing results,
namely, the assumption
that the input $\Gamma$-tree $T$ will never change. Indeed, if $T$ is updated,
these results must discard the index~$J$ and re-run the preprocessing phase on the new tree. 
To improve on this, we want our enumeration algorithm to support \emph{update
operations} on~$T$, and to update~$J$ accordingly
instead of recomputing it from scratch.
Specifically, an algorithm for \emph{enumeration under updates}
on a
tree $T$ has a
preprocessing phase that produces the index $J$ as usual, but has two algorithms
during the enumeration phase: (i.) an enumeration algorithm~$\mathcal{A}$ as
presented before, and (ii.) an \emph{update algorithm}~$\mathcal{U}$. 
When we want to change the tree~$T$, we call $\mathcal{U}$ with a
description of the changes:
$\mathcal{U}$ modifies~$T$ accordingly,
updates the index~$J$, and resets
the enumeration state (so enumeration starts over on the new tree, and all
working memory of the enumeration phase is freed). The
\emph{update time}
of the enumeration algorithm is the 
complexity of~$\mathcal{U}$: like preprocessing, but unlike delay,
it is a function of the size of the (current) tree~$T$.

To our knowledge, the only published result on enumeration for MSO queries under
updates is the work of
Losemann and Martens
\cite{losemann2014mso}, which applies to words and to trees, for MSO queries with only
free first-order variables. They show 
an enumeration algorithm
with linear-time preprocessing: on words,
the update complexity and delay is $O(\log \card{T})$;
on trees, these complexities become
$O(\log^2 \card{T})$.
Thus the delay is worse than in the case without updates~\cite{bagan2006mso},
and in particular it is no longer independent from~$T$.

\subparagraph*{Main result.}
In this work, we show that enumeration under updates for MSO queries on trees can
be performed with a better complexity that matches the case without updates:
linear-time preprocessing, linear delay and memory (in the assignments),
and update time in $O(\log \card{T})$. This improves on the
bounds of~\cite{losemann2014mso} (and uses entirely different techniques).
However, in exchange for the better complexity, we only support a weaker update
language: we can change the labels of tree nodes, called a
\emph{relabeling}, but we cannot insert or delete leaf nodes
as in~\cite{losemann2014mso}, which we leave
for future work (see the conclusion
in Section~\ref{sec:conclusion}).
We show in Section~\ref{sec:applications} that
relabelings are still useful to derive results for some practical query
languages. 

Formally, a relabeling on a $\Gamma$-tree~$T$ is a pair of a node $n\in T$ and a label
$l\in \Gamma$. To apply it, we change the label $\lambda(n)$
of~$n$ by adding $l$ if $l \notin \lambda(n)$, and removing it if $l \in
\lambda(n)$. In other words, the tree~$T$ never changes, and updates only
modify~$\lambda$.
Our main result is then:

\begin{theoremrep}
  \label{thm:main}
  For any fixed tree alphabet $\Gamma$ and MSO query $Q(\mathbf{X})$ on
  $\Gamma$-trees, given a $\Gamma$-tree $T$, we can enumerate the output $Q(T)$
  of~$Q$ on~$T$ with linear-time preprocessing, linear delay and memory, and
  logarithmic update time for relabelings.
\end{theoremrep}

\begin{proof}
  See Appendix~\ref{apx:together} for the proof of this result.
\end{proof}

\noindent In other words, after preprocessing $T$ in time~$O(\card{T})$ to compute the
index~$J$, we can:
\begin{itemize}
  \item Enumerate the assignments of~$Q$ on~$T$, using~$J$,
    with delay linear in the size of each assignment, so
    constant if the assignments to $Q$ have constant size.
  \item Toggle a label of a node of~$T$,
    update~$J$, and reset
    the enumeration, in time $O(\log \card{T})$.
\end{itemize}
We show this result in Sections~\ref{sec:provenance}--\ref{sec:reachability},
and then give consequences of this result in
Section~\ref{sec:applications}.

\section{Provenance Circuits}
\label{sec:provenance}
Our general technique for enumeration follows our earlier work~\cite{amarilli2017circuit}:
from the query and input tree, we compute in linear time a structure called a
\emph{provenance circuit}
to represent the results to enumerate,
we observe that it falls in a restricted circuit class,
and we conclude by showing a
general enumeration result
for circuits of this class. In this section, we review our construction of
provenance circuits in~\cite{amarilli2017circuit}, with some additional
observations
that will be useful for updates. In particular, we show an independent
\emph{balancing lemma} on input trees, which allows us to bound a parameter
of the circuit called \emph{dependency size}.
We will extend the formalism of this section to 
so-called \emph{hybrid circuits} in the next section; and we will show our enumeration result for such circuits
in Sections~\ref{sec:enumeration} and~\ref{sec:reachability}.

\subparagraph*{Set circuits.}
We start with some preliminaries about circuits.
A \emph{circuit} $C = (G, W, g_0, \mu)$ is a directed acyclic
graph $(G, W)$ whose vertices $G$ are called \emph{gates}, whose edges $W$
are called \emph{wires}, where $g_0 \in G$ is the \emph{output gate}, and
where $\mu$ is a function giving a \emph{type} to each gate of~$G$ (the possible
types depend on the kind of circuit).
The \emph{inputs} to a gate $g\in G$ are $\inp(g)
\colonequals \{g' \in G \mid (g',
g) \in W\}$ and the \emph{fan-in} of~$g$ is its number of inputs
$\card{\inp(g)}$.

We define \emph{set-valued circuits}, which are an equivalent rephrasing of
the \emph{circuits in
zero-suppressed semantics} used in~\cite{amarilli2017circuit}. They can also
be seen to be isomorphic to arithmetic circuits, and generalize
factorized representations used in database theory~\cite{olteanu2015size}.
The type function $\mu$ of a set-valued circuit maps each gate to one of $\aplus$,
$\atimes$, $\var$. We require 
that $\atimes$-gates have fan-in
0 or 2, and 
that $\var$-gates have fan-in~0: the latter are called
the \emph{variables} of~$C$, with $C_\var$ denoting the set of
variables.
Each gate $g$ of~$C$ \emph{captures} a set $\gset{g}$ of \emph{assignments}, where
each \emph{assignment} is a subset of~$C_\var$. These sets are
defined 
bottom-up 
as follows:
\begin{itemize}
  \item For a variable gate $g$, we have $\gset{g} \colonequals \{\{g\}\}$.
  \item For a $\aplus$-gate $g$, we have $\gset{g}
    \colonequals
    \bigcup_{g' \in \inp(g)} \gset{g'}$. In particular, if $\inp(g) = \emptyset$ 
    then $\gset{g} =
    \emptyset$.
  \item For a $\atimes$-gate $g$ with no inputs, we have $\gset{g} \colonequals
    \{\{\}\}$.
  \item For a $\atimes$-gate $g$ with two inputs $g_1$ and $g_2$, we have
    $\gset{g}
    \colonequals \{A_1 \cup A_2 \mid (A_1, A_2) \in
    \gset{g_1} \times \gset{g_2}\}$,
    which we write $\gset{g} \colonequals \gset{g_1} \relprod \gset{g_2}$ (this is the relational product).
\end{itemize}
The set $\gset{C}$ \emph{captured} by~$C$ is $\gset{g_0}$ for $g_0$ the output gate
of~$C$. Note that each assignment of~$\gset{C}$ is a satisfying assignment
of~$C$ when seen in the usual semantics of monotone circuits.

\subparagraph*{Structural requirements.}
Before defining our provenance circuits, we introduce some
structural restrictions that they will respect, and that will be useful for
enumeration.

The first requirement is that the circuit is a \emph{d-DNNF}. Our definition of
d-DNNF is inspired by~\cite{darwiche2001tractable} but
applies to set-valued circuits, as in~\cite{amarilli2017circuit} (see also the
z-st-d-DNNFs of~\cite{sugaya2017fast}).
For each gate $g$ of a set-valued circuit~$C$, we define the \emph{domain} $\dom(g)$ of~$g$ as the
variable gates having a directed path to~$g$. In particular, for~$g \in C_\var$, we have 
$\dom(g) = \{g\}$, and if $\inp(g) = \emptyset$ then $\dom(g) = \emptyset$.
We now call a $\atimes$-gate $g$ 
\emph{decomposable} if it has no inputs or if, letting $g_1' \neq g_2'$ be its
two inputs, the domains $\dom(g_1')$ and $\dom(g_2')$ are disjoint.
This ensures that no variable of~$C$ occurs both in an assignment of~$\gset{g_1'}$ and
in an assignment of~$\gset{g_2'}$.
We call a
$\aplus$-gate $g$ \emph{deterministic} if, for any two inputs
$g_1' \neq g_2'$
of~$g$, the sets $\gset{g_1'}$ and $\gset{g_2'}$
are disjoint, i.e., there is no assignment that occurs in both sets.
We call $C$ a \emph{d-DNNF} if every $\times$-gate is
decomposable and every $\aplus$-gate is deterministic.
This assumption allows us, e.g., to tractably compute the cardinality of the set $\gset{C}$ captured
by~$C$.

The second requirement on circuits is called \emph{upwards-determinism} and
was introduced in~\cite{amarilli2017circuit_extended}. In that paper, it was
used to show an improved memory bound;
in the present paper, we will always be able to enforce it.
A wire $(g, g')$ in a set-valued circuit $C$ 
is called \emph{pure} if:
\begin{itemize}
  \item $g'$ is a $\aplus$-gate; or
  \item $g'$
is a $\atimes$-gate and, letting $g''$ be the other input of~$g'$, we have
    $\{\} \in \gset{g''}$, i.e., $g''$ captures the empty assignment.
\end{itemize}
We say that a gate~$g$ is \emph{upwards-deterministic} if
there is at most one gate $g'$ such that $(g, g')$ is pure.
We call $C$ \emph{upwards-deterministic} if every gate of~$C$ is.

The third requirement concerns the \emph{maximal fan-in} of circuits, which is
simply defined for a set-valued circuit~$C$ as the maximal fan-in of a gate
of~$C$. We will require that the maximal fan-in is bounded by a constant.

The fourth and last requirement concerns a new parameter called \emph{dependency
size}. To introduce this, we define
the \emph{dependent gates} $\Delta(g')$ of a gate $g'$ in a set-valued circuit
$C$ as the gates
$g$ such that there is a directed path from~$g'$ to~$g$. Intuitively, the set
$\gset{g}$
captured by~$g$ may then depend on the set $\gset{g'}$ captured by~$g'$.
The \emph{dependency
size} of~$C$ is $\Delta(C) \colonequals \max_{g\in C} \card{\Delta(g)}$, i.e.,
the maximal number of gates that are dependent on any given gate~$g$. We will
require this parameter to be connected to the height of the input tree.

\subparagraph*{Set-valued provenance circuits.}
We can now define provenance circuits like in~\cite{amarilli2017circuit}. 
A set-valued circuit $C$ is a \emph{provenance circuit} of
a MSO query $Q(X_1, \ldots, X_m)$ on a~$\Gamma$-tree~$T$ if:
\begin{itemize}
  \item The variables of~$C$ correspond to the possible singletons, formally:
$C_\var = \{\langle
X_i :
    n\rangle \mid 1 \leq i \leq m \text{~and~} n \in T\}$; and
    \item The
set of assignments captured by~$C$ is the output of~$Q$ on~$T$, 
    formally: $\gset{C} = Q(T)$. Equivalently, 
for any tuple $\mathbf{B} = (B_1, \ldots, B_m)$ of subsets of~$T$, we have
$T \models Q(\mathbf{B})$ iff the assignment $\{\langle X_i : n\rangle \mid 1
\leq i \leq m \text{~and~} n
\in B_i\}$ is in~$\gset{C}$.
\end{itemize}

\begin{example}
  \label{exa:set}
  Consider the unlabeled tree~$T$ of Figure~\subref{fig:tree}, the alphabet
  $\Gamma = \{B\}$, and the MSO query $Q(x)$ with one free first-order
  variable asking for the leaf nodes
  whose $B$-annotation is different from that of its parent (i.e., the node
  carries label $B$ and the parent does not, or vice-versa). Consider the labeling
  $\lambda$ mapping $1$ to $\{B\}$ and $2$ and $3$ to~$\emptyset$. A
  set-valued circuit capturing the provenance of~$Q$ on~$(T, \lambda)$ is given in
  Figure~\subref{fig:set}.
\end{example}

We then know from~\cite{amarilli2017circuit_extended} that provenance circuits
can be computed efficiently, and they can be made to respect our structural
requirements:

\begin{theorem}[(from \cite{amarilli2017circuit}, Theorem~7.3)]
  \label{thm:provorig}
  For any fixed MSO query $Q(\mathbf{X})$ on~$\Gamma$-trees, given a
  $\Gamma$-tree $T$,
  we can compute in time $O(\card{T})$ a
  set-valued provenance circuit $C$ of~$Q$ on~$T$. Further, $C$ is a d-DNNF, it
  is upwards-deterministic, its maximal fan-in is constant, and its dependency
  size is in $O(\h(T))$, where $\h$ denotes the height of~$T$.
\end{theorem}

\begin{proofsketch}
  We recall the main proof technique: we convert $Q$ to a bottom-up
  deterministic tree automaton $A$ on $\Gamma$-trees, and we add nodes to~$T$
  to describe the possible valuations of variables. The provenance circuit
  $C$
  then captures the possible ways that $A$ can read~$T$
  depending on the valuation: we compute it with the construction
  of~\cite{amarilli2015provenance}, and is a d-DNNF thanks to automaton
  determinism (see~\cite{amarilli2016leveraging}).
  Upwards-determinism is shown like
  in~\cite{amarilli2017circuit_extended}.
  
  The bounds on fan-in and dependency size are not stated
  in~\cite{amarilli2017circuit,amarilli2017circuit_extended} but already hold 
  there. Specifically, the maximal fan-in is a function of the transition
  function of~$A$, i.e., it does not depend on~$T$.
  The bound on dependency size holds because $C$ is constructed
  following the structure of~$T$: we create for each tree node a gadget
  whose size depends only on~$A$, and we connect these
  gadgets precisely following the structure of~$T$, so that
  $\Delta(g)$ for any
  gate $g$ of~$C$ can only contain gates from the node~$n$ of~$g$ or from
  ancestors of~$n$ in the tree.
\end{proofsketch}

In the context of updates, the bound of dependency size will be crucial:
intuitively, it describes how many gates need to be updated when an update
operation modifies a gate of the circuit. As this bound depends on the
height of the input tree, we will conclude this section by a \emph{balancing
lemma} that ensures that this height can always be made logarithmic (which 
matches 
our desired update complexity). We will then add support for updates in the
next section by extending circuits to
\emph{hybrid circuits}.

\begin{toappendix}
  In this appendix, we prove Lemma~\ref{lem:balancing}:
\end{toappendix}

\subparagraph*{Balancing lemma.}
Our balancing lemma is a general observation on MSO query evaluation on trees,
and is in fact completely independent from provenance circuits.
It essentially says that the input tree can be assumed to be balanced. Formally, we will show that we can rewrite
any MSO query $Q$ on $\Gamma$-trees to an MSO query $Q'$
on a larger tree alphabet $\Gamma'$
so that any input tree $T$ for~$Q$ can be rewritten in linear time to a balanced tree
$T'$ on which~$Q'$ returns exactly the same output. Because we intend to
support update operations, the input tree $T$ will be unlabeled, and the
rewritten tree~$T'$ will work for any labeling of~$T$. Formally:

\begin{lemmarep}
  \label{lem:balancing}
  For any tree alphabet $\Gamma$ and MSO query $Q(\mathbf{X})$
  on~$\Gamma$-trees, we can compute a tree alphabet $\Gamma' \supseteq \Gamma$
  and MSO query $Q'(\mathbf{X})$ on~$\Gamma'$-trees such that the following
  holds. Given any unlabeled tree~$T$ with node set $N$, we can compute in
  linear time a $\Gamma'$-tree $(T', \lambda')$ with
  node set $N' \supseteq N$, such that 
  $\h(T') = O(\log \card{T})$
  and such that, for any labeling function $\lambda: T \to
  2^\Gamma$, we have $Q(\lambda(T)) = Q'(\lambda''(T'))$, where
  $\lambda''(n)$ maps $n \in T'$ to~$\lambda(n)$ if $n \in T$ and $\lambda'(n)$
  otherwise.
\end{lemmarep}

\begin{proofsketch}
  We prove Lemma~\ref{lem:balancing} by seeing the 
  input tree~$T$ as a relational structure~$I$ of
  treewidth~1, and invoking the result by
  Bodlaender~\cite{bodlaender1998parallel} to compute in linear
  time a constant-width tree decomposition of~$I$ which is of
  logarithmic height. We then translate the query $Q$ to a MSO query $Q'$ on
  tree encodings of this width, and compute from $T$ the tree encoding $T'$
  corresponding to the tree decomposition (we rename some nodes of~$T'$ to ensure that
  the nodes of~$T$ are reflected in~$T'$).
  Note that the balanced tree decompositions
  of~\cite{bodlaender1998parallel} were already used for similar
  purposes elsewhere, e.g., in~\cite{eppstein2017kbest}, end of Section~2.3.
\end{proofsketch}

\begin{toappendix}
  To prove Lemma~\ref{lem:balancing}, we will need to introduce preliminaries
about relational instances~\cite{abiteboul1995foundations}, tree decompositions,
and tree encodings.

\subparagraph*{Instances.}
A \emph{relational signature} is a set of \emph{relation names} 
together with an associated \emph{arity} (a non-zero natural number).
We fix a relational signature $\sigma_2$ that codes unlabeled trees, consisting 
of two binary relations $E_1$ and $E_2$ indicating the first and second
child of each internal node. For any tree alphabet $\Gamma$, we let
$\sigma_{\Gamma}$ denote a signature to represent labels of~$\Gamma$, i.e., one
unary relation $P_l$ for each $l \in \Gamma$. Last, for a tuple $\mathbf{X} =
X_1, \ldots, X_m$ of second-order variables, we let $\sigma_{\mathbf{X}}$ denote
a signature to represent the interpretation of these variables, i.e., one
unary relation $B_i$ for each $1 \leq i \leq m$. By \emph{monadic second-order
logic (MSO) over $\sigma$}, we denote MSO with the relations of~$\sigma$ and
equality in the usual way.

A \emph{relational instance} of a relational signature $\sigma$ is a set $I$ of
\emph{$\sigma$-facts} of the form $R(a_1, \ldots, a_n)$ where $a_1, \ldots, a_n$
are \emph{elements}, $R$ is a relation in~$\sigma$, and $n$ is the arity of~$R$.
The \emph{domain} $\dom(I)$ of~$I$ is the set of elements that occur in~$I$.

Given a $\Gamma$-tree $(T, \lambda)$, we can easily compute in linear time a couple $(I,
I_\lambda)$ where $I$ is a $\sigma_2$-instance describing the unlabeled tree~$T$
in the expected way (in particular, $\dom(I)$ is exactly the set of nodes
of~$T$), and $I_\lambda$ is the $\sigma_\Gamma$-instance
$\{P_l(n) \mid n \in T \text{~and~} l \in \lambda(n)\}$.

\subparagraph*{Tree decompositions.}
A \emph{tree decomposition} of an undirected graph $G = (V, E)$ is a tree
$\Theta$ (whose nodes are called \emph{bags}) 
and a labeling function $\dom:\Theta\to 2^V$ such
that:
\begin{itemize}
\item For every $e \in E$, there is $b \in \Theta$ such that $e \subseteq
\dom(b)$
\item For every $v \in V$, the set $T_v \colonequals \{b \in \Theta \mid v \in
\dom(b)\}$ is a connected subtree of~$\Theta$.
\end{itemize}
We still assume for convenience that tree decompositions are rooted, ordered,
binary, and full trees. Specifically, they will be computed as rooted binary
trees by \cite{bodlaender1998parallel}, they can be made full without loss of
generality (in linear time and without impacting the height) by adding empty
bags, and we can add an arbitrary order on the children of each internal bag to
make them ordered.
The \emph{width} of~$\Theta$ is $\max_{b \in \Theta} \card{\dom(b)} - 1$, and
the \emph{treewidth} of~$G$ is the smallest width of a tree decomposition
of~$G$.

A \emph{tree decomposition} of a relational instance $I$ is a tree decomposition
of its \emph{Gaifman graph}, i.e., the graph on vertex set $\dom(I)$ where there
is an edge between any two elements $a, a'$ that occur together in some fact. The
\emph{treewidth} of~$I$ is that of its Gaifman graph.

The definition of tree decompositions ensures that, for any relational instance
$I$ and tree decomposition $\Theta$, for any $a \in \dom(I)$, we can talk of the
topmost bag $b$ of~$\Theta$ such that $a \in \dom(b)$; we write this bag $\node(a)$.
This mapping $\node$ can be computed explicitly in linear time
given $I$
and $\Theta$ by~\cite[Lemma~3.1]{flum2002query}.

We will make a standard assumption on our tree decompositions, namely, that the
function $\node$ is an injective function: in other words, the root bag contains only one
element, and for any non-root bag $b$ with parent bag $b'$, we have
$\card{\dom(b) \setminus \dom(b')} \leq 1$. This requirement
can be enforced on a tree decomposition $\Theta$
in linear time using standard techniques, without impacting the width
of~$\Theta$, and only multiplying the height of~$\Theta$ by a constant (assuming
that the width is constant): specifically, we replace each bag violating the
condition by a chain of bags where the new elements are introduced one after the
other. Hence, we will always make this assumption.

We now recall the result of Bodlaender~\cite{bodlaender1998parallel}, which is
the key to our construction:

\begin{theorem}[from \cite{bodlaender1998parallel}]
  \label{thm:bodlaender}
  For any relational signature $\sigma$,
  given a relational instance~$I$ on~$\sigma$
  of width~$w \in \NN$, we can
compute in linear time in~$I$ a tree decomposition $\Theta$ of~$I$ of
  width~$O(w)$,
  such that $\h(\Theta)$ is in~$O(\log(\card{I}))$.
\end{theorem}

Specifically, the algorithm of~\cite{bodlaender1998parallel} is described for a
parallel machine, but can be run sequentially in linear time, as explained
in~\cite{eppstein2017kbest}, end of Section~2.3.

\subparagraph*{Tree encodings.}
If we fix a relational signature $\sigma$ and a treewidth bound $k\in\NN$, we
can compute an alphabet $\Gamma^k_\sigma$, called the \emph{alphabet of tree
encodings for~$\sigma$ and~$k$}, which ensures the following: given any
$\sigma$-instance $I$ with a tree decomposition $\Theta$ of width~$k$, we
can translate $I$ and $\Theta$ in linear time to a $\Gamma^k_\sigma$-tree $E$ (called a
\emph{tree encoding} of~$I$) that can be decoded back in linear time to an
instance isomorphic to~$I$. What is more, Boolean MSO formulas on
$\sigma$-instances (i.e., MSO formulas without free variables) can be
translated to Boolean MSO formulas on $\Gamma^k_\sigma$-trees that are
equivalent through encoding and decoding. An example of such a
scheme is given in~\cite{flum2002query}; we will use a different scheme,
detailed in~\cite{amarilli2016leveraging}, which ensures a property dubbed
\emph{subinstance-compatibility}: intuitively, removing a fact $F$ from~$I$ amounts
to toggling labels on a node of the tree encoding that corresponds to~$F$
(without changing the skeleton of the tree encoding).
The labels of $\Gamma^k_\sigma$ intuitively consist of a pair comprising a
\emph{domain}, i.e., a subset of elements among $2k+2$ fixed element names, and an
optional \emph{fact} on the elements of the domain.
We omit the formal definition of $\Gamma^k_\sigma$; see Section~3.2.1
of~\cite{amarilli2016leveraging} for details.

We are now ready to conclude the proof of Lemma~\ref{lem:balancing}:

\begin{proof}[Proof of Lemma~\ref{lem:balancing}]
  Let $Q(\mathbf{X})$ be the input query on $\Gamma$-trees.
  Let $\sigma \colonequals \sigma_2 \cup \sigma_\Gamma \cup \sigma_{\mathbf{X}}$.
  We let $Q'$ be the Boolean MSO query on $\sigma$-instances obtained from $Q$
  in the expected way, making it Boolean by replacing each second-order variable
  $X_i$ with the unary relation $B_i$ of~$\sigma_{\mathbf{X}}$. Given an input tree~$T$, we compute in
  linear time the $\sigma_2$-instance $I$ which represents it.
  It is clear that, given a labeling $\lambda:T\to 2^\Gamma$,
  recalling our earlier definition of the $\sigma_\Gamma$-instance ~$I_\lambda$
  from~$(T, \lambda)$,
  the output $Q(\lambda(T))$ of~$Q$ on~$\lambda(T)$ is equal to the set of
  $\sigma_{\mathbf{X}}$-instances $I'$ of $B_i$-facts on~$\dom(I)$ (seeing each
  such instance $I'$ as a set of singletons of the form $\langle X_i: a \rangle$) such that
  $I \cup I_\lambda \cup I'$ satisfies $Q'$.

  Let $w$ be the width of the tree decomposition obtained when applying
  Theorem~\ref{thm:bodlaender} to an input tree decomposition of width~$1$ (note
  that we have not specified the input yet).
  Let us compute from~$Q'$ the Boolean MSO query $Q''$
  on the alphabet~$\Gamma^w_{\sigma}$ of tree encodings for width~$w$
  which is equivalent to~$Q'$ on $\sigma$-instances (up to encoding and
  decoding), i.e., an instance on~$\sigma$ satisfies $Q'$ iff its encoding
  as a $\Gamma^w_{\sigma}$-tree satisfies $Q''$. We take $\Gamma'$ to consist
  of~$\Gamma^w_{\sigma}$ plus a special label $\nfix$, to be used later.

  Now, as~$T$ is a tree, the treewidth of~$I$ 
  is~$1$. Let us define an instance $I^+$ by adding to~$I$
  the instance $I^\Gamma$ of all possible $\sigma_\lambda$-facts on
  $\dom(I)$, plus the instance $I^{\mathbf{X}}$ of all possible
  $\sigma_{\mathbf{X}}$-facts on~$\dom(I)$. As all these additional facts
  are unary, the instance $I^+$ still has treewidth~$1$.
  Hence, by Theorem~\ref{thm:bodlaender},
  we can compute in linear time
  in~$I_+$ a tree decomposition $\Theta$ of~$I_+$ of treewidth~$w$ and logarithmic
  height.
  We also compute in linear
  time the mapping $\node:I_+\to\Theta$, and a tree encoding $E$ of~$I_+$, i.e., a
  $\Gamma^w_{\sigma}$-tree.
  
  Thanks to subinstance-compatibility, we know that, for any
  labeling $\lambda:T\to 2^\Gamma$ and answer tuple $\mathbf{B}$ of subsets of~$I$,
  letting
  $I_\lambda \subseteq I^\Gamma$ and $I_{\mathbf{B}} \subseteq I^{\mathbf{X}}$
  be the $\sigma_\Gamma$- and
  $\sigma_{\mathbf{X}}$-instances that respectively denote it, then we can obtain
  a tree encoding of $I \cup I_\lambda \cup I_{\mathbf{B}} \subseteq I^+$ by toggling the
  labels of some nodes of~$E$. Specifically, each fact of $I^\Gamma \cup
  I^{\mathbf{X}}$ corresponds to one node of~$E$ whose label has to be changed;
  further, this mapping can be computed in linear time
  (see~\cite{amarilli2016leveraging}, Lemma~3.2.6).

  The last thing to argue is that we can rename the nodes of~$E$ so that they
  correspond to the nodes of~$T$ associated to them, ensuring that, given a
  labeling function $\lambda:T\to 2^\Gamma$ of the tree~$T$, we can use it to
  relabel~$E$. (This
  differs slightly from the original construction
  of~\cite{amarilli2016leveraging}, because we want each node of~$T$ to be
  associated to one single node in~$E$, carrying all possible variables and
  labels; by contrast, in the construction of~\cite{amarilli2016leveraging},
  every fact corresponds to a specific node of~$E$.)
  To fix this, we
  modify $E$ in linear time to another $\Gamma^k_\sigma$-tree $E'$:
  for each $n \in \dom(I)$, letting $b \colonequals
  \node(a)$, we replace $b$ by a gadget with two copies $b_1$ and $b_2$ of $b$,
  with $b_2$ being the left child of~$b_1$. The label of~$b_1$ is that of~$b$,
  and the label of~$b_2$ is made of the same domain as~$b$ but without any fact;
  see the exact definition of $\Gamma^k_\sigma$ in Section~3.2.1
  of~\cite{amarilli2016leveraging} for details.
  We then add a right
  child to~$b_1$ which is a new node~$n$ 
  identified to the element~$n$ in~$\dom(I)$, 
  which itself corresponds to the node $n$ in~$T$; the label of~$b_1$ is the
  fixed special label~$\nfix$. 
  This construction is well-defined because the function $\node$ is injective.
  We must now argue that the query $Q''$ can
  be modified
  (independently from~$I$)
  to a MSO query $Q'''$ on $(\Gamma \cup \Gamma^k_{\sigma})$-trees
  to read labels and variable assignments from these new nodes: 
  specifically, instead of
  reading (the encodings of) the $\sigma_\Gamma$-facts about (the encoding of)
  an element $n \in \dom(I)$,
  the query $Q'''$ should read 
  the label in~$\Gamma$ of the new node of~$E'$ identified with~$n$; likewise,
  instead of reading (the encodings of) the $\sigma_{\mathbf{X}}$-facts on an
  element $n \in \dom(I)$ directly
  from~$E$, the query should read the
  $\mathbf{X}$-annotation of this same new node in~$E'$ identified with~$n$.
  To do this, the translations of the atoms from
  $\sigma_\Gamma$ and $\sigma_{\mathbf{X}}$ in~$Q''$ are replaced in~$Q'''$
  by a gadget
  which finds the bag where the corresponding element was introduced (i.e., the
  one for which it is in the image of~$\node$), finds the new node that we
  added with label~$\nfix$, and reads the label and annotation of this node. We also add a conjunct
  to~$Q'''$ to assert that the only nodes that can be part of the interpretation
  of the~$\mathbf{X}$ are the new nodes in~$E'$ with label~$\nfix$, thus ensuring that the set
  of answers of~$Q'''$ on any labeling $\lambda''(E')$ of~$E'$ is correct. 
  This concludes the proof.
\end{proof}

\end{toappendix}

\section{Hybrid Circuits for Updates}
\label{sec:hybrid}
In this section, we extend set-valued circuits to support updates, defining
\emph{hybrid circuits}. We then extend Theorem~\ref{thm:provorig} for these
circuits. Last, we introduce a new structural notion of \emph{homogenization}
of hybrid circuits
and show how to enforce it. We close the section
by stating our main enumeration result on hybrid circuits, which implies 
our main theorem (Theorem~\ref{thm:main}), and is proved in the two next
sections.

\subparagraph*{Hybrid circuits.}
A hybrid circuit is intuitively similar to a set-valued circuit, but it
additionally 
has \emph{Boolean variables} (which can be toggled when updating),
Boolean gates ($\land$, $\lor$, $\neg$), and gates labeled~$\ctimes$ which
keep or discard a set of assignments depending on a Boolean value.
Formally, a \emph{hybrid circuit} $C = (G, W, g_0, \mu)$ is a circuit where
the possible
gate types are
$\avar$ (\emph{set-valued variables}),
$\bvar$ (\emph{Boolean variables}),
$\aplus$, $\atimes$, $\ctimes$,
$\land$, $\lor$, and $\neg$.
We call a gate
\emph{Boolean} if its type is $\bvar$, $\land$, $\lor$, or~$\neg$; and \emph{set-valued}
otherwise. We require that the output gate $g_0$ is set-valued and that the following
conditions hold:
\begin{itemize}
  \item $\avar$-gates and $\bvar$-gates have fan-in exactly 0;
  \item All inputs to $\land$-gates, $\lor$-gates, and $\neg$-gates are Boolean,
    and $\neg$-gates have fan-in exactly~$1$;
  \item All inputs to $\aplus$ and $\atimes$-gates are set-valued, and $\atimes$-gates have fan-in either 0 or 2;
  \item $\ctimes$-gates have one set-valued input and one
    Boolean input (so they have fan-in exactly~2).
\end{itemize}
We write $C_\bvar$ to denote the gates of~$C$ of type $\bvar$, called the
\emph{Boolean variables}, and define likewise the \emph{set-valued variables}
$C_\avar$. An example hybrid circuit is illustrated in Figure~\subref{fig:hybrid}.

\tikzstyle{treenode}=[draw,circle]
\tikzstyle{agate}=[draw,circle,inner sep=1.5pt]
\tikzstyle{avar}=[draw,circle,inner sep=-1pt]
\tikzstyle{bvar}=[draw,inner sep=1.5pt]
\tikzstyle{neggate}=[draw,inner sep=.5pt]

\begin{figure}
    \centering
    \begin{subfigure}[t]{.145\textwidth}
        \centering
        \begin{tikzpicture}[yscale=.8,xscale=.5]
        \node[treenode] (a) at (0, 0) {1};
        \node[treenode] (b) at (-1, -1) {2};
        \node[treenode] (c) at (1, -1) {3};
        \path[draw] (a) -- (b);
        \path[draw] (a) -- (c);
        \end{tikzpicture}
        \caption{Example\\\mbox{unlabeled tree}}
        \label{fig:tree}
    \end{subfigure}\hfill
    \begin{subfigure}[t]{.145\textwidth}
        \centering
        \begin{tikzpicture}[yscale=.8,xscale=.45]
        \node[agate] (a) at (0, 0) {$\aplus$};
        \node[avar] (b) at (-1, -1) {\sgx{2}};
        \node[avar] (c) at (1, -1) {\sgx{3}};
        \path[draw] (a) -- (b);
        \path[draw] (a) -- (c);
        \end{tikzpicture}
        \caption{Example\\\null~~~set circuit}
        \label{fig:set}
    \end{subfigure}\hfill
    \begin{subfigure}[t]{0.58\textwidth}
      \begin{tikzpicture}[xscale=1.85,yscale=.65]
          \node[agate] (a) at (0, -.9) {$\aplus$};
          \node[agate] (b1) at (-1, -1.4) {$\ctimes$};
          \node[agate] (b2) at (1, -1.4) {$\ctimes$};
          \node[bvar] (c1) at (-1.5, -2.5) {\sgb{1}};
          \node[agate] (c2) at (-.5, -1.95) {$\aplus$};
          \node[neggate] (c3n) at (.775, -1.895) {$\bm{\neg}$};
          \node[bvar] (c3) at (.5, -2.5) {\sgb{1}};
          \node[agate] (c4) at (1.5, -1.95) {$\aplus$};
          \node[agate] (d1) at (-1, -2.5) {$\ctimes$};
          \node[agate] (d2) at (0, -2.5) {$\ctimes$};
          \node[agate] (d3) at (1, -2.5) {$\ctimes$};
          \node[agate] (d4) at (2, -2.5) {$\ctimes$};
          \node[neggate] (e1n) at (-1.115, -3.1) {$\bm{\neg}$};
          \node[bvar] (e1) at (-1.25, -3.7) {\sgb{2}};
          \node[avar] (e2) at (-.75, -3.7) {\sgx{2}};
          \node[neggate] (e3n) at (-.115, -3.1) {$\bm{\neg}$};
          \node[bvar] (e3) at (-.25, -3.7) {\sgb{3}};
          \node[avar] (e4) at (.25, -3.7) {\sgx{3}};
          \node[bvar] (e5) at (.75, -3.7) {\sgb{2}};
          \node[avar] (e6) at (1.25, -3.7) {\sgx{2}};
          \node[bvar] (e7) at (1.75, -3.7) {\sgb{3}};
          \node[avar] (e8) at (2.25, -3.7) {\sgx{3}};
          \path[draw] (a) -- (b1);
          \path[draw] (a) -- (b2);
          \path[draw] (b1) -- (c1);
          \path[draw] (b1) -- (c2);
          \path[draw] (b2) -- (c3n);
          \path[draw] (b2) -- (c4);
          \path[draw] (c3n) -- (c3);
          \path[draw] (c2) -- (d1);
          \path[draw] (c2) -- (d2);
          \path[draw] (c4) -- (d3);
          \path[draw] (c4) -- (d4);
          \path[draw] (d1) -- (e1n);
          \path[draw] (d1) -- (e2);
          \path[draw] (d2) -- (e3n);
          \path[draw] (d2) -- (e4);
          \path[draw] (d3) -- (e5);
          \path[draw] (d3) -- (e6);
          \path[draw] (d4) -- (e7);
          \path[draw] (d4) -- (e8);
          \path[draw] (e1n) -- (e1);
          \path[draw] (e3n) -- (e3);
        \end{tikzpicture}
        \caption{Example hybrid circuit.
        Boolean gates are squared,\\
        set-valued gates are
        circled, and variables are
        repeated}
        \label{fig:hybrid}
    \end{subfigure}\hfill
    \begin{subfigure}[t]{.125\textwidth}
    \centering
      \begin{tikzpicture}[yscale=.8,xscale=.5]
      \node[agate] (a) at (0, 0) {$\aplus$};
      \node[agate] (b1) at (-1, -.8) {$\aplus$};
      \node[agate] (b2) at (1, -.8) {$\aplus$};
        \node[avar] (c1) at (-1, -2) {\sgx{2}};
        \node[avar] (c2) at (1, -2) {\sgx{3}};
      \path[draw] (a) -- (b1);
      \path[draw] (a) -- (b2);
      \path[draw] (b1) -- (c1);
      \path[draw] (b1) -- (c2);
      \path[draw] (b2) -- (c1);
      \path[draw] (b2) -- (c2);
      \end{tikzpicture}
      \caption{Example\\switchboard}
      \label{fig:switchboard}
    \end{subfigure}
\end{figure}

Unlike set-valued circuits, which capture only one set of assignments, hybrid
circuits capture several different sets of assignments, depending on the value
of the Boolean variables (intuitively corresponding to the tree labels).
This value is given by
a \emph{valuation} of~$C$, i.e., a function
$\nu:C_\bvar\to\{0,1\}$. Given such a valuation~$\nu$,
each Boolean gate $g$ \emph{captures} a Boolean value $\bval{\nu}{g} \in \{0, 1\}$,
computed bottom-up in the usual way: we set $\bval{\nu}{g} \colonequals \nu(g)$
for $g \in C_{\bvar}$, and otherwise $\bval{\nu}{g}$ is
the result of the Boolean operation given by the type $\mu(g)$ of~$g$, applied
to the Boolean values $\bval{\nu}{g'}$ captured by the inputs $g'$ of~$g$
(in particular, a $\land$-gate with no
inputs always has value~$1$, and a $\lor$-gate with no inputs always has
value~$0$).

We then define the \emph{evaluation} of~$C$ under~$\nu$ as the
set-valued circuit $\nu(C)$ obtained as follows.
First, replace each Boolean gate $g$ of~$C$ by a
$\atimes$-gate with no inputs (capturing~$\{\{\}\}$) 
if $\bval{\nu}{g}=1$, and by a $\aplus$-gate with no
inputs (capturing~$\emptyset$) if $\bval{\nu}{g} = 0$. Second, 
relabel each $\ctimes$-gate $g$ of~$C$ to be 
a $\atimes$-gate.
Using $\nu(C)$, for each set-valued gate~$g$ of~$C$, we define the set
\emph{captured by~$g$ under~$\nu$}: it is the set of
assignments (subsets of~$C_{\avar}$) that $g$
captures in~$\nu(C)$.
The set $\gsetv{\nu}{C}$
\emph{captured} by~$C$ under~$\nu$ is then~$\gsetv{\nu}{g_0}$,
for~$g_0$ the output gate of~$C$.

We last lift the structural definitions from set-valued circuits to hybrid
circuits. The \emph{maximal fan-in} and \emph{dependency size} of a hybrid
circuit are defined like before (these definitions do not depend
on the kind of circuit). A hybrid circuit $C$ is a \emph{d-DNNF},
resp.\ is \emph{upwards-deterministic}, if for every valuation $\nu$ of~$C$, the
set-valued circuit $\nu(C)$ has the same property. For instance, 
the hybrid circuit in Figure~\subref{fig:hybrid} is upwards-deterministic and is
a d-DNNF.

\subparagraph*{Hybrid provenance circuits.}
We can now use hybrid circuits to define provenance with support for updates.
The set-valued
variables of the circuit will correspond to singletons as before, describing
the interpretation of the \emph{free variables} of the query;
and the Boolean
variables stand for a different kind of singletons, describing which
\emph{labels}
are carried by each node.
To describe this formally, we will consider
an \emph{unlabeled} tree~$T$, and define a \emph{labeling assignment} of~$T$ for a
tree alphabet~$\Gamma$ as a set of singletons of the form $\langle l : n \rangle$
where $l \in \Gamma$ and $n \in T$. Given a labeling assignment $\alpha$, we can define
a labeling function $\lambda_\alpha$ for~$T$, which maps each node $n\in T$ to
$\lambda(n) \colonequals \{l \in \Gamma \mid \langle l: n\rangle \in \alpha\}$.
Now, we say that a hybrid
circuit~$C$ is a \emph{provenance circuit} of a MSO query $Q(X_1, \ldots, X_m)$
on an \emph{unlabeled} tree~$T$ if:
\begin{itemize}
  \item The set-valued variables of~$C$ correspond
to the possible singletons in an assignment, formally $C_{\avar} = \{\langle X_i : n
    \rangle \mid 1 \leq i \leq m \text{~and~} n \in T \}$;
\item The Boolean variables of~$C$ correspond to the possible singletons in a
  labeling assignment, formally $C_{\bvar} = \{\langle l : n\rangle \mid 
    l \in \Gamma \text{~and~} n \in T\}$;
  \item For any labeling assignment $\alpha$, let $\nu_\alpha$ be the Boolean valuation
    of~$C_{\bvar}$ mapping each $\langle l : n\rangle$ to~$0$ or~$1$ depending
    on whether $\langle l : n\rangle \in \alpha$ or not, and let $\lambda_\alpha$ be the
    labeling function on~$T$ defined as above. Then we require that the set of
    assignments $\gsetv{\nu_\alpha}{C}$ captured by~$C$ under~$\nu_\alpha$ is
    exactly the output of~$Q$ on~$\lambda_\alpha(T)$, formally, $\gsetv{\nu_\alpha}{C}
    = Q(\lambda_\alpha(T))$.
\end{itemize}
In other words, for each labeling $\lambda$ of the tree~$T$, considering the
valuation $\nu$ that sets the Boolean
variables of~$C$ accordingly, then $\nu(C)$ is a provenance circuit
for~$Q$ on~$\lambda(T)$.

\begin{example}
  Recall the query $Q(x)$ and alphabet $\Gamma = \{B\}$ of Example~\ref{exa:set}, and the tree~$T$ of
  Figure~\subref{fig:tree}. A hybrid circuit $C$ capturing the provenance of~$Q$ on~$T$
  is given in Figure~\subref{fig:hybrid} (with variable gates being drawn at
  multiple places for legibility): square leaves correspond to Boolean
  variables testing node labels, and
  circle leaves correspond to set-valued variables capturing a singleton
  of the form \sgx{n} for some $n \in T$.
  In particular, for
  the labeling~$\lambda$ of Example~\ref{exa:set},
  the corresponding valuation $\nu$ maps \sgb{1} to~$1$ and \sgb{2} and \sgb{3}
  to~$0$, and 
  the evaluation $\nu(C)$ of~$C$ under~$\nu$ captures the same set as the
  circuit of Figure~\subref{fig:set}.
\end{example}

We can now extend Theorem~\ref{thm:provorig} to compute a hybrid provenance
circuit as follows:

\begin{toappendix}
  \subsection{Proof of the Provenance Circuit Theorem}
  \label{apx:prfprovenance}
  In this appendix, we prove Theorem~\ref{thm:provenance}:
\end{toappendix}

\begin{theoremrep}
  \label{thm:provenance}
  For any fixed MSO query $Q(\mathbf{X})$ on $\Gamma$-trees, given an unlabeled
  tree $T$, we can compute in time $O(\card{T})$ a hybrid provenance circuit $C$
  which is a d-DNNF, is upwards-deterministic, has constant maximal fan-in, and
  has dependency size in $O(\h(T))$.
\end{theoremrep}

\begin{proofsketch}
  The proof is analogous to that of Theorem~\ref{thm:provorig}. The only
  difference is that the automaton now reads the label of each node as if it
  were a variable, so that the provenance circuit $C$ also reflects these label
  choices as Boolean variables.
\end{proofsketch}

\begin{toappendix}
  The general idea is that,
given the MSO query $Q(\mathbf{X})$ on $\Gamma$-trees, 
writing $\mathbf{X} = X_1, \ldots, X_m$, 
we define a query
$Q'(\mathbf{X}, \mathbf{Y})$ on unlabeled trees, where $\card{\mathbf{Y}} =
\Gamma$, with one second-order variable $Y_l$ corresponding to each $l \in
\Gamma$. The construction is simply that we replace each unary predicate $P_l$
in~$Q$ by the corresponding second-order variable $Y_l$. It is now obvious that,
for any labeled tree $(T,\lambda)$, defining $D_l \colonequals \{n \in T \mid
l \in \lambda(n)\}$ for each $l \in \Gamma$,
for any set $\mathbf{B} = B_1, \ldots, B_m$ of subsets of~$T$, 
we have $(T, \lambda) \models Q(\mathbf{B})$ iff $T \models Q(\mathbf{B},
\mathbf{D})$. In other words, we have simply turned node labels into
second-order variables.

Now, at a high level, we can simply construct a provenance circuit of~$Q'$
on~$T$ in the sense of Theorem~\ref{thm:provorig}, replace the input gates
corresponding to the $Y_j$ variables by a Boolean input gate, and observe that
the desired properties hold. We will now give a self-contained proof of the
construction, to make sure that we reflect the changes in definitions between
the present work and~\cite{amarilli2017circuit,amarilli2017circuit_extended}.

\subparagraph*{Tree automata.}
We will need to introduce some prerequisites about tree automata.
Given a tree alphabet $\Lambda$,
a \emph{bottom-up deterministic tree automaton} on~$\Lambda$, or
\emph{$\Lambda$-bDTA}, is a tuple $A = (Q, F, \iota, \delta)$ where $Q$ is a
finite set of \emph{states}, $F \subseteq Q$ are the \emph{final} states, $\iota
: \Lambda \to Q$ is the \emph{initial function}, and $\delta : Q^2 \times
\Lambda
\to Q$ is the \emph{transition function}. The \emph{run} of a $\Lambda$-bDTA~$A$ on a
$\Lambda$-tree $T$ is the function $\rho : T \to Q$ defined inductively as
$\rho(n) \colonequals \iota(\lambda(n))$ when $n$ is a leaf, and $\rho(n)
\colonequals \delta(\rho(n_1), \rho(n_2), \lambda(n))$ when $n$ is an internal node
with children $l_1$ and $l_2$. We say that $A$ \emph{accepts} the
tree~$T$ if the run $\rho$ of~$A$ on~$T$ maps the root of~$T$ to a final state.

We will be interested in bDTAs to capture our non-Boolean query $Q'$ on
unlabeled trees. Let $\mathbf{Z} \colonequals \mathbf{X} \cup \mathbf{Y}$
be the set of \emph{variables}, and let $\Lambda \colonequals 2^{\mathbf{Z}}$,
where $2^{\mathbf{Z}}$ denotes
the powerset of~$\mathbf{Z}$.
Letting $T$ be an unlabeled tree, we call a
\emph{$\mathbf{Z}$-annotation} of~$T$ a function $\nu: T \to 2^{\mathbf{Z}}$: the annotation intuitively describes the
interpretation of the variables of~$\mathbf{Z}$ by annotating each node with the
set of variables to which it belongs.
Letting $A$ be a $2^{\mathbf{Z}}$-bDTA, $T$ be an
unlabeled tree, and $\nu$ be a $\mathbf{Z}$-annotation of~$T$, we say that $\nu$ is
a \emph{satisfying annotation} of~$A$ on~$T$ if $A$ accepts $\nu(T)$. In this
case, we see $\nu$ as defining an \emph{assignment} $\alpha_\nu$, which is
the set
$\{\langle Z_i: n \rangle \mid 1 \leq i \leq \card{Z} \text{~and~} n \in T \}$.
The \emph{output} of $A$ on~$T$, written $A(T)$, is the set of assignments
corresponding to its satisfying annotations.
Following Thatcher and Wright~\cite{thatcher1968generalized}, and determinizing the automaton
using standard techniques \cite{tata}, the output of an MSO query
(here, on an unlabeled tree)
can be computed as the output of an automata for that query.
Formally:

\begin{lemma}[\cite{thatcher1968generalized,tata}]
  \label{lem:automata}
  Given a MSO query $Q(\mathbf{Z})$ on unlabeled trees, we can compute a
  $2^{\mathbf{Z}}$-bDTA $A$ such that, for any unlabeled tree $T$, we have
  $Q(T) = A(T)$.
\end{lemma}

\subparagraph*{Restricting to Boolean annotations.}
It will be more convenient in the sequel to assume that each tree node carries
one
single Boolean annotation rather than many, and to distinguish the annotations
corresponding to~$\mathbf{X}$ (the original variables of~$Q$, called
\emph{enumerable}), and those corresponding to~$\mathbf{Y}$ (the labels of the input
tree, called \emph{updatable}).
We will do this by creating 
$\card{\mathbf{Z}}$-copies of each tree node~$n$, to stand for each separate
singleton $\langle Z_i: n \rangle$.
To do this, we will consider the fixed alphabet $\Sigmafix = \{\enu, \upd,
\fix\}$.
Intuitively, $\enu$ will be the label of nodes whose annotation corresponds
to a variable of~$\mathbf{X}$, $\upd$ will be the label of nodes whose annotation
corresponds to a variable of~$\mathbf{Y}$, and $\fix$ will be the label of nodes whose
annotation does not code any variable and should be ignored.
Given a $\Sigmafix$-tree $T$, we will write $T_\enu$, $T_\upd$, and $T_\fix$ to
refer to the set of nodes carrying each label.
We will then consider 
$\overline{\Sigmafix}$-trees, where $\overline{\Sigmafix} \colonequals \Sigmafix \times
\{0, 1\}$, the alphabet of $\Sigmafix$-trees annotated with a Boolean value at
each node: as
promised, each node carries one single value.
Now, a \emph{Boolean annotation} of a $\Sigmafix$-tree $T$ is a function $\nu:T\to\{0,
1\}$, and we see $\nu(T)$ as a $\overline{\Sigmafix}$-tree defined in the expected
way.

We want to rephrase the evaluation of~$A$ on an unlabeled tree~$T$
to a problem on $\Sigmafix$-trees, where variable valuations are
coded in Boolean annotations. This process is formalized in the following lemma,
whose construction is illustrated in Figure~\ref{fig:tobool}; it is analogous to
Lemma~E.2 of~\cite{amarilli2017circuit}:

\begin{lemmarep}
  \label{lem:tobool}
  For any variable set~$\mathbf{Z} = \mathbf{X} \cup \mathbf{Y}$,
  given a $2^{\mathbf{Z}}$-bDTA $A$,
  we can compute a $\overline{\Sigmafix}$-bDTA $A'$ such that the following
  holds: given an unlabeled tree $T$, we can compute in linear time 
  a $\Sigmafix$ tree $T'$ of height $O(h(T))$
  and an injective function $\phi:T\times \mathbf{Z} \to T'$
  such that:
  \begin{itemize}
    \item $T'_\enu$ is exactly the set of nodes $n'$ such that $\phi(n, X_i) =
    n'$ for some $n \in T$ and $X_i \in \mathbf{X}$;
    \item $T'_\upd$ is exactly the set of nodes $n'$ such that $\phi(n, Y_j) =
    n'$ for some $n \in T$ and $Y_j \in \mathbf{Y}$;
    \item $T'_\fix$ is exactly the set of nodes not in the image of~$\phi$, and
      it includes all internal nodes.
  \end{itemize}
  Further, for any
  $\mathbf{Z}$-annotation
  $\nu: T \to 2^{\mathbf{Z}}$, let
  $\nu'$ be the Boolean valuation of~$T'$ defined by:
  \begin{itemize}
    \item If $n'\in T'$ is in the image of~$\phi$, then letting $(n, Z_i)
      \colonequals \phi^{-1}(n')$, we set $\nu'(n') \colonequals 1$ iff $Z_i \in
      \nu(n)$;
    \item If $n'\in T'$ is not in the image of~$\phi$, we set $\nu'(n')
      \colonequals 0$.
  \end{itemize}
  Then $A$ accepts $\nu(T)$ iff $A'$ accepts $\nu'(T')$.
\end{lemmarep}

\begin{figure}
  \begin{tikzpicture}[baseline=(current bounding box.center), level
    distance=.7cm]
    \Tree [.{$n: \{X_1, Y_2\}$} 
      [.{$n_1: \ldots$} ]
      [.{$n_2: \ldots$} ]
      ];
  \end{tikzpicture}
  $\Rightarrow$\hspace{-4em}
  {
  \footnotesize
  \begin{tikzpicture}[baseline=(current bounding box.center), level
    distance=.7cm]
    \Tree [.{$n: \fix, 0$} 
      [.{$n':\fix, 0$} 
        [.{$n''':\fix, 0$} 
          [.{$\phi(n,X_1): \enu, 1$} ]
          [.{$\phi(n,X_2): \enu, 0$} ]
        ]
        [.{$n'''':\fix, 0$} 
          [.{$\phi(n,Y_1): \upd, 0$} ]
          [.{$\phi(n,Y_2): \upd, 1$} ]
        ]
      ]
      [.{$n'': \fix, 0$}
        [.{$n_1: \ldots$} ]
        [.{$n_2: \ldots$} ]
      ]
      ];
  \end{tikzpicture}
  }
  \caption{Illustration of the construction of Lemma~\ref{lem:tobool}, with
  $\mathbf{Z} = (X_1, X_2, Y_1, Y_2)$, on an internal leaf~$n$ of a tree,
  and for a valuation mapping $n$ to $\{X_1, Y_2\}$.
  The
  node $n$ is replaced by a gadget of nodes with fixed labels $\fix$. Two
  $\enu$-labeled descendants indicate the annotation for the $X_i$,
  and two $\upd$-labeled descendants indicate the annotation for the
  $Y_i$.}
  \label{fig:tobool}
\end{figure}

\begin{proof}
  Given an input tree $T$, we change it following the idea of
  Figure~\ref{fig:tobool}: we replace each node $n$ by a gadget of nodes
  labeled with $\fix$, having two subtrees: one whose leaves are labeled $\enu$
  and code the variables $X_i$ in order, and another whose leaves are labeled
  $\upd$ and code the variables $Y_j$ in order.
  This gadget can be completed to a full binary tree by
  adding leaves labeled $\fix$ as necessary. Now we can clearly
  rewrite the $2^{\mathbf{Z}}$-bDTA to a $\overline{\Sigmafix}$-bDTA
  $A'$ which is equivalent in the sense required by the lemma. The
  states of $A'$ consist of the states of~$A$, the pairs of states
  of~$A$, and \emph{annotation} states which consist of binary sequences of
  length up to~$\card{\mathbf{Z}}$.
  The final states are the final states of~$A$. The
  initial function $\iota'$ and transition function $\delta'$ are informally coded as follows.
  The initial function $\iota'$ maps nodes labeled $(\enu, b)$ or $(\fix, b)$ for $b \in
  \{0, 1\}$ 
  to the singleton binary sequence $(b)$ formed of its Boolean value, and it
  maps nodes labeled $(\fix, b)$ for $b \in \{0, 1\}$ to the empty binary
  sequence. The transition function $\delta'$ is defined only on nodes labeled $(\fix, b)$
  for $b \in \{0, 1\}$, because all internal nodes of~$T'$ carry such a label
  (as required); and it is defined as follows (where we ignore the Boolean
  annotation~$b$ of the node):
\begin{itemize}
  \item Given two states $q_1$ and $q_2$ of~$A$, the new state is the pair
    $(q_1, q_2)$;
  \item Given two states that are binary sequences of length $<
    \card{\mathbf{Z}}$, the new state is their concatenation;
  \item Given a binary sequence $s$ of length $\card{\mathbf{Z}}$ and a pair of
    states $(q_1, q_2)$, the new state is the state $\delta(q_1, q_2, s)$
    of~$A$, where $\delta$ is the transition function of~$A$;
  \item Given a binary sequence $s$ of length $\card{\mathbf{Z}}$ and an empty
    binary sequence, the new state is the state $\iota(s)$.
\end{itemize}
  On Figure~\ref{fig:tobool}, the automaton~$A'$ would reach state $(1, 0)$
  on~$n'''$, reach state $(0, 1)$ on~$n''''$ and reach state $(1, 0, 0, 1)$
  on~$n'$. Letting $q_1$ and $q_2$ be the states that $A'$ reaches
  respectively on~$n_1$ and~$n_2$, it reaches state $(q_1, q_2)$ on~$n''$.
  Hence, on node~$n$, it reaches $\delta(q_1, q_2, (1, 0, 0, 1))$. This figure
  illustrates the translation when~$n$ is an internal node with children~$n_1$
  and~$n_2$. The case where~$n$ is a leaf is described in the last bullet point,
  and is analogous: the leaf $n$ in~$T$ is translated to a node $n$ in~$T'$ with one left child
    $n'$ that is the root of the tree describing the valuation of~$\mathbf{Z}$,
    and one right child $n''$ labeled $\fix$ which is a leaf of~$T'$.

  Now, it is easy to show that $A'$ is equivalent to~$A$ in the sense of the
  lemma statement, which concludes the proof.
\end{proof}

We now have a $\overline{\Sigmafix}$-bDTA~$A'$ to run on a $\Sigmafix$-tree $T$. We
can now rephrase our desired provenance result
as a provenance result on such automata. We say that a
hybrid circuit $C$ is a \emph{provenance circuit} of a $\overline{\Sigmafix}$-bDTA
$A'$ on a $\Sigmafix$-tree $T$ if:

\begin{itemize}
  \item The set-valued variables of~$C$ correspond to the nodes of~$T$ with
  label $\enu$, formally, $C_{\avar} = T_\enu$
  \item The Boolean variables of~$C$ correspond to the nodes of~$T$ with label
  $\upd$, formally, $C_{\bvar} = T_\upd$
  \item For any Boolean valuation $\nu$ of~$T$ such that $\nu(n) = 0$ for each
  $n\in T_\fix$, the automaton $A'$ accepts $\nu(T)$ iff, letting $\nu'$ be
  the restriction of~$\nu$ to~$T_\upd$, and letting $B$ be
  the set of nodes of~$T$ corresponding to the restriction of~$\nu$ to~$T_\enu$,
  we have $B \in \nu'(C)$.
\end{itemize}

We can now rephrase our desired result. Note that the statement of this result
implies that our construction is also tractable in the automaton, as we
mentioned in the conclusion (Section~\ref{sec:conclusion}):

\begin{theorem}
  \label{thm:provauto}
  Given a $\overline{\Sigmafix}$-bDTA $A'$
  and a $\Sigmafix$-tree $T$ where all internal nodes are labeled~$\fix$,
  we can compute in time $O(\card{T} \times \card{A'})$
  a hybrid circuit $C$ which is a provenance circuit of~$A'$ on~$T$. Further, $C$
  is a d-DNNF, it is upwards-deterministic, its maximal fan-in is in
  $O(\card{A'})$, and its dependency size is in $O(\card{A'} \times \h(T))$, where
  $h(T)$ is the height of~$T$.
\end{theorem}

\begin{proof}
  We adapt the proof of Proposition~E.8
  from~\cite{amarilli2017circuit_extended}.
  Let us write the $\overline{\Sigmafix}$-bDTA $A' = (Q, F, \iota, \delta)$. 
  We will construct the circuit $C$ in
  a bottom-up fashion from~$T$. We consider every node $n$ of~$T$ with label
  $\lambda(n) \in \Sigmafix$.

  If $n$ is a leaf node, for $b \in \{0, 1\}$, we let $q_b \colonequals
  \iota((\lambda(n), b))$, and we create the following gates in~$C$:
  \begin{itemize}
    \item One set-valued gate~$g_n$ and one set-valued gate $g_{\neg n}$, defined as
      follows:
      \begin{itemize}
        \item If $n \in T_\fix$, the gate $g_n$ is a $\aplus$-gate with no
          inputs (i.e., the annotation at~$n$ is always~$0$), and the gate
          $g_{\neg n}$ is a $\atimes$-gate with no inputs;
        \item If $n \in T_\upd$, the gate $g_n$ is a $\ctimes$-gate of
          a $\atimes$-gate with no inputs and of a Boolean variable gate identified
          to~$n$; and the gate $g_{\neg n}$ is a $\ctimes$-gate of 
          a $\atimes$-gate with no inputs and of the negation of the Boolean
          variable gate previously mentioned;
        \item If $n \in T_\enu$, the gate $g_n$ is a set-valued input gate
          identified to~$n$; and the gate $g_{\neg n}$ is a $\atimes$-gate with
          no inputs.
      \end{itemize}
    \item One $\aplus$-gate $g^q_n$ for each state $q\in Q$ with the following
      inputs:
      \begin{itemize}
        \item If $q = q_1$, the gate $g_n$;
        \item If $q = q_0$, the gate $g_{\neg n}$.
      \end{itemize}
      In particular, if $q \neq q_0$ and $q \neq q_1$, then $g^q_n$ is an
      $\aplus$-gate with no inputs, and if $q_0 = q_1$ then the gate $g^{q_0}_n$
      has both inputs.
  \end{itemize}

  If $n$ is an internal node with children $n_1$ and $n_2$, remembering that
  necessarily $n \in T_\fix$, we create the following gates in~$C$:

  \begin{itemize}
    \item One $\atimes$-gate $g_n^{q_1, q_2}$ for each $q_1, q_2 \in Q$ whose
      inputs are $g^{q_1}_{n_1}$ and $g^{q_2}_{n_2}$;
    \item One $\aplus$-gate $g_n^q$ for each $q \in Q$ whose inputs are all the
      gates~$g^{q_1,q_2}_n$ such that we have $\delta((\lambda(n), 0), q_1, q_2)
      = q$. In particular, if there are no states $q_1, q_2$ such that this
      equality holds, then $g_n^q$ is an $\aplus$-gate with no inputs.
  \end{itemize}

  The output gate $g_0$ is a $\aplus$-gate of the $g^q_{n_\r}$ for all $q \in
  F$, where $n_\r$ is the root of~$T$.

  It is clear that the construction satisfies the requirements of a hybrid
  circuit. It is also clear that the construction obeys the prescribed time
  bounds. The only gates in the construction whose arities are not obviously
  bounded
  are the
  $\aplus$-gates, and they always have at most $O(\card{A'})$-transitions (the
  bound is in the size of the transition table of~$A'$), so the fan-in bound is
  respected. For the dependency size, if we consider an arbitrary gate $g$ of
  the circuit, let $n$ be the node of~$T$ for which it was created. It is clear
  that $\Delta(g)$ is a subset of the set of all gates created for a node $n'$
  which is an ancestor of~$T$ in~$n$. Now, we create $O(\card{A'})$-gates for
  each node of~$T$, so indeed the dependency size is bounded by $O(\card{A'}
  \times h(T))$.
  We must now show that the circuit has the correct semantics, that it is a
  d-DNNF, and that it is upwards-deterministic.

  \medskip
  
  We first show that the
  semantics of the circuit is correct,
  by showing by bottom-up induction the invariant that for any valuation $\nu$
  of~$C$, for all $n \in T$, letting $T_n$ be the subtree of~$T$ rooted at~$n$,
  the set $\gsetv{\nu}{g^q_n}$ precisely denotes
  the set of assignments of $T_\enu \cap T_n$ such that the following holds:
  letting $\nu^n_\enu:
  T_\enu \cap T_n
  \to \{0, 1\}$ be the Boolean function corresponding to the assignment,
  letting $\nu^n_\upd$ be the
  restriction of~$\nu$ to~$T_n \cap T_\upd$ (remember that the domain of~$\nu$ is
  the Boolean gates of~$C$, i.e., $T_\upd$),
  letting $\nu^n_\fix:T_\fix \cap T_n \to \{0, 1\}$ be the constant-$0$ function,
  and letting $\nu^n\colonequals
  \nu^n_\enu
  \cup \nu^n_\upd \cup\nu^n_\fix$ be the valuation of~$T_n$ defined
  from~$\nu^n_\enu$, $\nu^n_\upd$, and $\nu^n_\fix$ in the expected way,
  the automaton $A'_q$ accepts $\nu^n(T_n)$,
  where $A'_q$ is the $\overline{\Sigmafix}$-bDTA obtained from $A'$ by setting
  $q$ as the only final state. This set of assignments is denoted $A'_q(\nu,
  T_n)$ in what follows.

  For the base case of a leaf $n \in T_\fix$, we know that $A'_q(\nu, T_n)$ is the
  empty set for all $q \neq \iota((\lambda(n), 0))$, and that it is the set
  $\{\{\}\}$ otherwise; this is what our construction ensures.

  For the base case of a leaf $n \in T_\upd$, we know that, for all $b \in \{0,
  1\}$, if $\nu(n) = b$, then $A'_q(\nu, T_n)$ is the empty set for all $q \neq
  \iota((\lambda(n), b))$, and that it is the set $\{\{\}\}$ otherwise; again,
  this is exactly what we ensure.

  For the base case of a leaf $n \in T_\enu$, letting $q_b \colonequals
  \iota((\lambda(n), b))$ for all $b \in \{0, 1\}$, we know that
  $A'_{q_0}(\nu, T_n)$
  contains the empty assignment $\{\}$, and $A'_{q_1}(\nu, T_n)$ contains the
  singleton assignment $\{n\}$ (if $q_0 = q_1$ then $A'_{q_0}(\nu, T_n)$ contains both),
  and that is all: this is what we ensure.

  For the induction case, letting $n$ be an internal node of~$T$ with children
  $n_1$ and $n_2$, assuming by induction hypothesis that
  $\gsetv{\nu}{g^{q_1}_{n_1}} = A'_{q_1}(\nu, T_{n_1})$ and
  $\gsetv{\nu}{g^{q_2}_{n_2}} = A'_{q_2}(\nu, T_{n_2})$ for all $q_1, q_2 \in Q$, we know by
  definition that, for all $q \in Q$, the output $A'_q(\nu, T_n)$ consists of the union,
  for $q_1, q_2 \in Q$ such that $\delta((\lambda(n), 0), q_1, q_2) = q$, of the
  relational product of the outputs $A'_{q_1}(\nu, T_{n_1})$
  and $A'_{q_2}(\nu, T_{n_2})$. This is
  because is a bijection between the Boolean labelings of~$T_n$ that are
  accepted by~$A'_q$, and the pairs of Boolean labelings of~$T_{n_1}$ and
  of~$T_{n_2}$ that are respectively accepted by~$A'_{q_1}$ and~$A'_{q_2}$ for
  some $q_1, q_2$ satisfying the condition. Again, this is precisely what we
  compute, so we have shown the invariant.

  Now, as $A'(\nu, T) = \bigcup_{q \in F} A'_q(\nu, T)$, we have established that
  $\gsetv{\nu}{C}$ is correct.

  \medskip

  We now show that $C$ is a d-DNNF. For decomposability, we will show a slightly
  stronger property. Remember that, in the main text, we said that~$C$ is a d-DNNF
  if, for any valuation~$\nu$ of~$C$, the set-valued circuit $\nu(C)$ is a
  d-DNNF. We will instead define decomposability directly on the hybrid
  circuit~$C$. Define the function $\dom$ on~$C$ as follows: for any set-valued gate~$g$ of~$C$, we denote by
  $\dom(g)$ the set of set-valued variable gates having a directed path to~$g$
  in~$C$. We now say that an $\atimes$-gate $g$ of~$C$ is \emph{decomposable} if it has
  no inputs or if, letting $g_1$ and~$g_2$ be its two inputs, the sets
  $\dom(g_1)$ and~$\dom(g_2)$ are disjoint. We then call~$C$ \emph{decomposable} if
  this holds for every $\atimes$-gate $g$ of~$C$. Let us show that $C$ is
  decomposable in this sense, which clearly implies that
  $\nu(C)$ is decomposable for every valuation~$\nu$ of~$C$.
  Now, the only
  $\atimes$-gates with two inputs are the $g^{q_1,q_2}_n$, whose inputs are
  $g^{q_1}_{n_1}$ and $g^{q_2}_{n_2}$ for the two children $n_1, n_2$ of~$n$.
  Now, an immediate bottom-up induction shows that for any node $n'$ of~$T$ and
  state~$q$, we have $\dom(g^q_n) \subseteq \{n' \in T_\enu \cap T_n\}$. Hence, indeed, the domains are disjoint.

  For determinism, the $\aplus$-gates created for the leaf nodes $n \in T_\upd$ may have two
  inputs, but in this case, the sets that they capture are clearly disjoint for
  any valuation~$\nu$, because one is always empty depending on~$\nu(n)$. For
  the $\aplus$-gates $g^q_n$ created for a state $q$ of~$A'$ and an internal node
  $n$ of~$T$ with children $n_1$ and~$n_2$, assume by contradiction that there
  is a valuation $\nu$ of~$C$ and some assignment~$a$ such that
  $a \in \gsetv{\nu}{g^{q_1,q_2}_n}$  and
  $a \in \gsetv{\nu}{g^{q_1',q_2'}_n}$  for $(q_1,q_2) \neq (q_1',q_2')$. Assume that
  $q_1 \neq q_1'$, the case $q_2 \neq q_2'$ is analogous.
  By our inductive invariant and the construction of the circuit, 
  and by the definition of the output of automata,
  we know that for the valuation $\nu'$ of~$T_n$ defined from~$a$ and~$\nu$, 
  the automata $A'_{q_1}$ and $A'_{q_1'}$ both accept $\nu'(T_{n_1})$. This
  contradicts the determinism of~$A'$, so we have a contradiction. Hence, $g^q_n$ is
  deterministic. The last gate to consider is the output gate~$g_0$, but if it
  has two different inputs $g^{q_1}_{n_\r}$ and $g^{q_2}_{n_\r}$ (for~$n_\r$ the
  root of~$T$) such that some assignment~$a$ belongs both to
  $\gsetv{\nu}{g^{q_1}_{n_\r}}$ and to
  $\gsetv{\nu}{g^{q_2}_{n_\r}}$, then the inductive invariant shows that
  the automata $A'_{q_1}$ and $A'_{q_2}$ both accept $\nu'(T)$, for~$\nu'$ the
  valuation defined from~$a$ and~$\nu$, contradicting again the determinism of~$A'$.
  We have thus shown that $C$ is a d-DNNF.

  \medskip

  We must last show that $C$ is upwards-deterministic. We adapt the argument of
  Claim~F.3 of~\cite{amarilli2017circuit_extended}. As $A'$ is deterministic,
  each gate of the form $g^{q_1,q_2}_n$ is used as input to only one gate
  $g^q_n$, namely, the one defined according to the transition function; and all
  set-valued gates introduced at leaves of~$T$ are also used as inputs to only
  one gate, except the~$g^q_n$. So the
  only set-valued gates in the construction which are used as inputs to multiple
  gates are the $g^q_n$ when $n$ is a leaf of~$T$ or an internal node of~$T$
  which is not the root. Fix a valuation $\nu$
  of~$C$. Let $n'$ be the parent of~$n$, and assume that $n = n_1$ is the first
  child of~$n'$; the other case is symmetric. Let $n_2$ be the other child
  of~$n'$. Now, the gate $g^q_n$ is used as inputs to gates of the form
  $g^{q,q_2}_{n'}$ for~$q_2 \in Q$, and the other input of these gates is $g^{q_2}_{n_2}$. Let $\nu'$ be the
  extension of~$\nu$ obtained by labeling all nodes of~$T_\enu \cap T_n$ and
  of~$T_\fix \cap T_n$ with~$0$: it is
  a valuation of~$T$. Now, by determinism of~$A'$, we know that there is exactly
  one state $q'$ such that $A'_{q'}$ accepts $\nu'(T_{n_2})$. Hence, by our inductive
  invariant, the only $g^{q_2}_{n_2}$ such that $\{\} \in
  \gsetv{\nu}{g^{q_2}_{n_2}}$ is
  $g^{q'}_{n_2}$; so for~$\nu$ there is exactly one pure outgoing wire
  connecting $g^q_n$ to another gate, namely, the one connecting it to
  $g^{q,q'}_{n'}$. This shows that $g^q_n$ is upwards-deterministic. Hence, we
  have shown that $C$ is upwards-deterministic. This concludes the proof.
\end{proof}

We can now recap the proof of Theorem~\ref{thm:provenance}:

\begin{proof}[Proof of Theorem~\ref{thm:provenance}]
Given the MSO query $Q(\mathbf{X})$ on $\Gamma$-trees, 
writing $\mathbf{X} = X_1, \ldots, X_m$, 
define a query
$Q'(\mathbf{X}, \mathbf{Y})$ on unlabeled trees as we explained initially, and
write $\mathbf{Z} = \mathbf{X} \cup \mathbf{Y}$. Use Lemma~\ref{lem:automata} to
compute a $2^{\mathbf{Z}}$-bDTA $A$ such that $Q'(T) = A(T)$. Now, use
Lemma~\ref{lem:tobool} to compute the $\overline{\Sigmafix}$-bDTA $A'$.
All of this is independent from the input tree.

Now, when we are given the unlabeled tree~$T$ as input, we compute in linear
time the $\Sigmafix$-tree~$T'$ and the injective function~$\phi$
  described in Lemma~\ref{lem:tobool}. Now, we use
Theorem~\ref{thm:provauto} to compute a provenance circuit $C'$ of~$A'$ on~$T'$.
We know that $C'$ is a d-DNNF, that it is upwards-deterministic, and that its
maximal fan-in depends only on~$A'$, so it is constant. Further, its dependency
size is in $O(\h(T') \times \card{A'})$, i.e., it is in $O(\h(T))$ because
$\h(T') = O(\h(T))$. Now, let us relabel the inputs of~$C'$: for every $g \in
C'_{\avar}$, remembering that it corresponds to a node $n' \in T'_\enu$, letting
$(n, X_i) \colonequals \phi^{-1}(n')$, we relabel $g$ to~$\langle X_i:
n\rangle$. We also relabel every $g \in C_{\bvar}$ to $\langle Y_j: n\rangle$ in the
same way. Let $C$ be the result of this renaming on~$C'$ ; as this
renaming is bijective, $C$ is still an upwards-deterministic d-DNNF and the
  dependency size and maximal fan-in is unchanged.

  To show that the circuit $C$ is correct,
  remember that a labeling assignment $\alpha$ of~$T$
  is a set of singletons $\langle l: n\rangle$ with $l \in \Gamma$, 
  which we can see as a set of pairs $\langle Y_i: n\rangle$ because
  $\mathbf{Y}$ corresponds to~$\Gamma$.
  We must show that for every labeling assignment $\alpha_\upd$,
  letting $\nu_{\alpha_\upd}$ be the Boolean valuation of~$C_{\bvar}$ defined
  from~$\alpha_\upd$, and $\lambda_{\alpha_\upd}$ be the $\mathbf{Y}$-annotation of~$T$ defined
  from~$\alpha_\upd$, then the set of assignments captured by~$C$ under
  $\nu_{\alpha_\upd}$
  is exactly the output of~$Q$ on~$\lambda_{\alpha_\upd}(T)$.
  Let $\alpha_\upd$ be such a labeling assignment, and let us show the claim.
  Let $\alpha_\upd'$ be the subset of~$T'_\upd$ obtained as the image
  of~$\alpha_\upd$ via
  the mapping~$\phi$ of Lemma~\ref{lem:tobool}, i.e., $\alpha_\upd' \colonequals \{n'
  \in T'_\upd \mid \langle Y_j: n \rangle \in \alpha_\upd \text{~where~}
  (n, Y_j) \colonequals \phi^{-1}(n')\}$. 
  Let~$\nu'_{\alpha_\upd'}$ be the Boolean valuation of~$T'_\upd$ defined
  from the subset~$\alpha_\upd'$ of~$T'_\upd$. 
  By definition of $C'$ being a provenance circuit of~$A'$, we know that
  $\nu'_{\alpha_\upd'}(C')$
  is exactly the set of subsets $\alpha_\enu'$ of~$T'_\enu$ such that,
  letting $\nu'_{\alpha_\enu'}$ be the Boolean valuation of~$T'_\enu$ corresponding
  to the subset~$\alpha_\enu'$ of~$T'_\enu$, 
  letting $\nu'_\fix$ be the Boolean valuation of~$T'_\fix$ mapping every node to~$0$,
  letting $\nu' \colonequals \nu'_{\alpha_\upd'} \cup \nu'_{\alpha'_\enu} \cup \nu'_\fix$
  the automaton $A'$ accepts $\nu'(T')$. 
  This last condition is equivalent, by the statement of Lemma~\ref{lem:tobool},
  to saying that~$A$ accepts $\lambda_{\alpha_\upd,\alpha_\enu}(T)$, which maps
  every $n \in T$ to the union of $\lambda_{\alpha_\upd}(n)$ and of
  $\lambda_{\alpha_\enu}(n) \colonequals \{X_i \in \mathbf{X} \mid
  \nu_{\alpha'_\enu}(\phi(n, X_i)) = 1\}$.
  This is equivalent to saying that $A$ accepts the result of
  annotating~$\lambda_{\alpha_\upd}(T)$ by the $\mathbf{X}$-annotation
  $\lambda_{\alpha_\enu}$, so by
  definition of~$A$ it is equivalent to saying that $\alpha_\enu$
  is an answer to~$Q$ on~$\lambda_{\alpha_\upd}(T)$.
  So, to summarize, we know that $\nu'_{\alpha_\upd}(C')$ is exactly the set of
  subsets $\alpha_{\enu'}$ of~$T'_\enu$ such that the corresponding $\alpha_{\enu}$ is in
  the output of~$Q$ on~$\lambda_{\alpha_\upd}(T)$.
  Thanks to the renaming that we performed from~$C'$ to~$C$, we know that
  $\nu'_{\alpha_\upd}(C)$ is exactly the output of~$Q$ on~$\lambda_{\alpha_\upd}(T)$, which
  establishes correctness, and concludes the proof.
\end{proof}

\end{toappendix}

\begin{toappendix}
  \subsection{Proof of the Homogenization Lemma}
  \label{apx:prfhomogenize}
  In this appendix, we prove Lemma~\ref{lem:homogenize}:
\end{toappendix}

\subparagraph*{Homogenization.}
We will make enumeration simpler by 
imposing one last requirement on hybrid circuits.
A hybrid circuit $C$ is \emph{homogenized} if there is no valuation
$\nu$ of~$C$ and set-valued gate $g$ of~$C$ such that $\{\} \in \gsetv{\nu}{g}$.
Note that the requirement does not apply to the Boolean gates of~$C$, nor to
the gates that replace them in evaluations~$\nu(C)$ of~$C$, so it equivalently
means that $C$ does not contain $\atimes$-gates with no inputs.
Intuitively, set-valued gates in~$C$ that capture the empty assignment would waste time in
the enumeration.
We will show that we can rewrite circuits in
linear time to make them homogenized, while preserving our requirements; but
we need to change our definitions slightly to ensure that the circuit
can still capture the empty assignment overall. To do so, we add the
possibility of distinguishing a Boolean gate $g_1$ of a hybrid circuit $C$ as
its \emph{secondary output};
in this case, given a
valuation $\nu$ of~$C$, the set $\gsetv{\nu}{C}$
captured by~$C$ under~$\nu$ is $\gsetv{\nu}{C}$ plus the empty assignment
$\{\}$ if the secondary output $g_1$ evaluates to~$1$, i.e., if $\bval{\nu}{g_1} = 1$.
We say that two hybrid circuits $C$ and $C'$ (with or without secondary outputs)
are \emph{equivalent} if
$C_\bvar = C'_\bvar$, $C_\avar = C'_\avar$, and for any valuation $\nu$
of~$C$, we have $\gsetv{\nu}{C} = \gsetv{\nu}{C'}$.
We then have:

\begin{lemmarep}
  \label{lem:homogenize}
  For any hybrid circuit $C$,
  we can build in linear time a hybrid circuit
  $C'$ with a secondary output $g_1$, such that $C'$ is homogenized and it is equivalent to~$C$. Further, if $C$ is a
  d-DNNF and is upwards-deterministic, then so is~$C'$; if $C$ has bounded fan-in
  then the same holds of~$C'$; and we have $\Delta(C') = O(\Delta(C))$.
\end{lemmarep}

\begin{proofsketch}
  This is shown analogously to homogenization
in~\cite{amarilli2017circuit}, which follows the
  technique of Strassen~\cite{Strassen73} (only done
  for two ``layers'', namely, empty and non-empty assignments).
\end{proofsketch}

\begin{proof}
  We first describe the construction. We will inductively rewrite each
  set-valued
  gate $g$ of~$C$ to two gates $g_{=0}$ and $g_{>0}$ of~$C'$, to preserve the
  following invariant.
  First, the gate $g_{>0}$ will be set-valued and ensure that, 
  for
  any valuation $\nu$ of~$C$, we have
  $\gsetv{\nu}{g_{>0}} = \gsetv{\nu}{g} \setminus \{\{\}\}$.
  Second, the gate $g_{=0}$ will be Boolean and ensure that, for any
  valuation $\nu$ of~$C$, we have $\bval{\nu}{g_{=0}} = 1$ iff $\{\} \in
  \gsetv{\nu}{g}$.
  We first copy all Boolean gates of~$C$ as-is in~$C'$, so in
  particular their evaluation is always the same in~$C$ and in~$C'$. The precise
  construction is then the following:
  \begin{itemize}
    \item For the base case on a set-valued variable gate $g$, we identify
  $g_{>0}$ to~$g$, and we let $g_{=0}$ be a $\lor$-gate with no inputs, so that
  $\bval{\nu}{g_{=0}} = 0$ for each valuation~$\nu$.

\item For the base case on a $\aplus$-gate $g$ with no inputs, we define $g_{=0}$
  to be a $\lor$-gate with no inputs, and let
  $g_{>0}$ be a $\aplus$-gate with no inputs.

\item For the base case on a $\atimes$-gate $g$ with no inputs, we define
  $g_{=0}$ to be a $\land$-gate with no inputs (so that it always evaluates
  to~$1$), and define $g_{>0}$ to be a $\aplus$-gate with no inputs.

\item For the induction case on a $\ctimes$-gate $g$, letting $g'$ be its set-valued
  input and $g''$ its Boolean input, letting $g'_{=0}$ and $g'_{>0}$ be the gates
  of~$C'$ obtained by induction for~$g'$, we define $g_{>0}$ as
      a $\ctimes$-gate of~$g''$ and $g'_{>0}$, and define $g_{=0}$ as 
      an $\land$-gate of~$g''$ and $g'_{=0}$.

\item For the induction case on a $\aplus$-gate $g$, let $g^1, \ldots, g^n$ be
  its inputs, with $n$ being a constant. We define $g_{=0}$ as an $\lor$-gate
  of~$g^1_{=0}, \ldots, g^n_{=0}$, and $g_{>0}$ as a $\aplus$-gate
  of~$g^1_{>0}, \ldots, g^n_{>0}$.

\item For the induction case on a $\atimes$-gate $g$, let $g^1$ and $g^2$ be its
  two inputs. We let $g_{=0}$ be a $\land$-gate of $g^1_{=0}$ and
  $g^2_{=0}$. We let $g_{>0}$ be a $\aplus$-gate of:
  \begin{itemize}
    \item a $\atimes$-gate of $g^1_{>0}$ and $g^2_{>0}$
    \item a $\ctimes$-gate of $g^1_{=0}$ and $g^2_{>0}$
    \item a $\ctimes$-gate of $g^1_{>0}$ and $g^2_{=0}$
  \end{itemize}
\end{itemize}
  Finally, we let the output gate $g_0$ of~$C'$ be $(g_0)_{>0}$, and let the secondary
  output gate $g_1$ of~$C'$ be $(g_0)_{=0}$.

  \medskip

  It is easy to check that $C'$ satisfies the conditions on hybrid circuits,
  and that the invariant is verified, so that $C'$ is indeed equivalent
  to~$C$. Further, the invariant ensures that no set-valued gate of~$C'$
  captures $\{\}$ under some valuation, so $C'$ is indeed homogenized.

  It is clear that the construction is in linear time.
  It is also clear that the maximum fan-in of~$C'$ is no bigger than that of~$C$
  (unless it is less than~$3$, in which case it is~$3$).
  Further, it is clear that if there is a directed path from a gate $g_1'$ to a
  gate $g_2'$ in~$C'$, then letting $g_1$ and $g_2$ be the gates from which
  $g_1$ and $g_1'$ were created, there is a directed path from~$g_1$ to~$g_2$
  in~$C$. As we create only constantly many gates in~$C'$ for each gate
  of~$C$, this ensures that the dependency of~$C'$ is at most that of~$C$
  multiplied by a constant.

  \medskip

  We now show that, if $C$ is a d-DNNF, then $C'$ also is. 
  Specifically, letting $\nu$ be a valuation of~$C$, we show that if~$\nu(C)$ is a
  d-DNNF then so is~$\nu(C')$.
  This is like in
  Proposition~B.3 of~\cite{amarilli2017circuit_extended} except it is
  simpler in our context because there are no range gates and the arity
  bounds are more convenient.
  The only $\atimes$-gates in~$\nu(C')$ are those 
  created in the
  first bullet point of the list for the last induction case: now as
  $\dom(g^1)$ and $\dom(g^2)$ are disjoint in~$\nu(C)$ because $g$ is
  decomposable, and as we clearly have by
  construction that $\dom(g^1_{>0}) = \dom(g_1)$ (identifying $g'$ and $g_{>0}'$
  for $g' \in C_\avar$) and likewise $\dom(g^2_{>0}) = \dom(g_2)$,
  we can conclude. For determinism, we need to consider
  first the induction case for $\aplus$, and second the induction case for
  $\atimes$. For $\aplus$, the determinism of the $\aplus$-gate $g_{>0}$
  follows from that of~$g$ in~$\nu(C)$. For $\atimes$, the gate $g_{>0}$ is indeed
  deterministic because:
  \begin{itemize}
    \item Letting $g'$ be its first input, as $C'$ is homogenized, each
      answer in $\gsetv{\nu}{g'}$  must contain one variable from $\dom(g^1_{>0})$
      and one from $\dom(g^2_{>0})$
    \item Letting $g''$ be its second input, the answers in $\gsetv{\nu}{g''}$
      contain only variables from $\dom(g^1_{>0})$
    \item Letting $g'''$ be its second input, the answers in $\gsetv{\nu}{g'''}$
      contain only variables from $\dom(g^2_{>0})$
  \end{itemize}
  We conclude that the answers are indeed disjoint. Hence, $g_{>0}$ is indeed
  deterministic, which establishes that $C'$ is a d-DNNF.

  \medskip

  Last, we show that if $C$ is upwards-deterministic then so is~$C'$;
  specifically, we show that for any valuation $\nu$ of~$C$, if $\nu(C)$ is
  upwards-deterministic then so is~$\nu(C')$. The
  proof is analogous to that of Claim~F.7
  in~\cite{amarilli2017circuit_extended}. First note that the gates added when
  evaluating $C'$ to~$\nu(C')$ cannot break upwards-determinism because they are
  used as the input to only one gate (the replacement of a $\ctimes$-gate), so
  we can ignore them. Like in the proof of~\cite{amarilli2017circuit_extended},
  we define the
  \emph{original gate} $\omega_\nu(g)$ of any of the other
  set-valued gates $g$ of~$\nu(C')$ to be the gate
  in~$\nu(C)$ for which it was created: specifically, if $g$ is of the form
  $(g')_{>0}$ then $\omega_\nu(g) \colonequals g'$, and if $g$ is a fresh
  gate created for a $\aplus$-gate $g'$ of~$\nu(C)$ in the last induction case
  above, then $\omega_\nu(g) \colonequals g'$. Clearly the gates
  in~$\nu(C')$ that come from fresh gates in~$C'$
  cannot violate upwards-determinism, because they are used as input to only
  one gate in~$C'$, hence in~$\nu(C')$.
  So it suffices to show that, for any gate $g$ of~$\nu(C)$, the
  gate $g_{>0}$ in~$\nu(C')$ is upwards-deterministic.
  We want to show that there is at most one gate $g'$ in~$\nu(C')$ such
  that the wire $(g_{>0}, g')$ is pure in~$\nu(C')$. First observe that, if $g'$ is a
  fresh gate from the first sub-item in the last induction case, then clearly this
  wire is not pure, because the other input is a set-valued gate and~$C$ is
  homogenized so the other input cannot capture the empty assignment (neither
  in~$C'$ nor in~$\nu(C')$).
  Hence, we can exclude these wires from consideration. Now,
  in fact, for
  any wire $w = (g_{>0}, g')$ in~$C'$, then
  $\omega_\nu(w) \colonequals (g, \omega_\nu(g'))$ is also a wire of~$\nu(C)$, and this mapping is
  injective: there are no two wires $w \neq w'$ such that $\omega_\nu(w) =
  \omega_\nu(w')$. Indeed, for each gate $g$ of~$C$ and outgoing wire of~$g$
  in~$\nu(C)$, we create at most one wire from $g_{>0}$ to a gate of~$\nu(C')$ among
  the wires that remain at this stage.
  Hence, using the
  upwards-determinism of~$g$ in~$\nu(C)$ 
  it suffices to show that whenever a wire $w$ of~$C'$ was not excluded yet and is
  pure, then $\omega_\nu(w)$ also is.

  To do so, we consider the possible wires $w = (g_{>0}, g')$ in~$\nu(C')$:
  \begin{itemize}
    \item If $g'$ is a $\aplus$-gate, then $\omega_\nu(g')$ also was, which
      concludes.
    \item If $g'$ is a $\atimes$-gate coming from a $\atimes$-gate
      of~$C'$, we have already excluded these wires.
    \item If $g'$ is the translation in~$\nu(C')$ (as a $\atimes$-gate)
      of a $\ctimes$-gate of~$C'$ created in the induction case
      for~$\ctimes$, then its other input is the translation of a Boolean gate
      which existed also in~$\nu(C)$ and had the same value,
      so the wire is pure iff the corresponding wire
      is pure in~$\nu(C)$
    \item If $G'$ is the translation in~$\nu(C')$ of a $\ctimes$-gate of~$C'$ created in the induction case
      for~$\atimes$, then its second input in~$\nu(C')$ captures the empty
      assignment iff it stands in~$C'$
      for a $\atimes$-gate of the form $g''_{=0}$ which
      evaluates to true under~$\nu$, i.e., iff the original gate~$g''$ captured the
      empty set in~$\nu(C)$, so again we have an equivalence.
  \end{itemize}
  This concludes the proof of preservation of upwards-determinism, and
  concludes the proof.
\end{proof}

Hence, up to linear-time processing, we can additionally assume that the
circuits of Theorem~\ref{thm:provenance} are homogenized. We can now use this
theorem, the lemma above, and Lemma~\ref{lem:balancing}, to reduce enumeration for MSO on trees (as in our main theorem,
Theorem~\ref{thm:main}) to the task of enumerating the set captured by a hybrid circuit
satisfying some structural properties. The result that we need is the following
(we prove it in the next two sections):

\begin{toappendix}
  \subsection{Proof of the Main Circuit Theorem}
  \label{apx:prfcircuits}
\end{toappendix}

\begin{theoremrep}
  \label{thm:circuits}
  Given an upwards-deterministic, d-DNNF, homogenized hybrid circuit~$C$ with
  constant fan-in, given an initial Boolean valuation $\nu$
  of~$C_{\bvar}$, there is an enumeration algorithm with linear-time preprocessing to
  enumerate the set $\gsetv{\nu}{C}$ captured by~$C$ under~$\nu$, with linear delay and
  memory in each produced assignment, and with update time 
  in~$O(\Delta(C))$: an \emph{update} consists here of toggling one value
  in~$\nu$.
\end{theoremrep}

\begin{proof}
  See Appendix~\ref{apx:together} for the proof of this result.
\end{proof}

\section{Enumerating Assignments of Hybrid Circuits}
\label{sec:enumeration}
In this section and the next, we 
prove Theorem~\ref{thm:circuits} by giving
an algorithm for enumeration under updates.
We start by describing the preprocessing
phase, computing two simple structures: a \emph{shortcut function} and a
\emph{partial evaluation}; we also explain how this index can be efficiently
updated. We then describe an algorithm for the enumeration phase, which needs an
additional index structure to achieve the required delay. We close the section by
presenting the missing index, called a \emph{switchboard}. The switchboard must
support a kind of reachability queries with a specific algorithm for enumeration
under updates: we give a self-contained presentation of this scheme in the next
section.

\subparagraph*{Preprocessing phase: shortcuts and partial evaluation.}
The first index structure that we precompute on our hybrid
circuit~$C$
consists of a \emph{shortcut function}
to avoid wasting time in chains of $\ctimes$-gates.
For each $\ctimes$-gate $g$, we precompute the one set-valued gate, called
$\delta(g)$ which is not
a $\ctimes$-gate and which has a directed path to~$g$ going only through
$\ctimes$-gates. 
The function $\delta$ can clearly be computed in a
linear-time bottom-up pass during the preprocessing, and it will never need to
be updated (it does not depend on~$\nu$). For notational convenience, we extend
$\delta$ by setting $\delta(g) \colonequals g$ for any set-valued gate~$g$ which
is not a $\ctimes$-gate.

The second index structure that we precompute is a \emph{partial
evaluation}, which depends on the
valuation $\nu$: it is a function $\omega_\nu$
from the gates of~$C$ to~$\{0, 1\}$ satisfying the following:
\begin{itemize}
  \item For every Boolean gate $g$, we have $\omega_\nu(g) = \bval{\nu}{g}$.
  \item For every set-valued gate $g$, we have $\omega_\nu(g) = 1$ iff
    $\gsetv{\nu}{g}$ is
    non-empty.
\end{itemize}
The function $\omega_\nu$ is intuitively an evaluation of the Boolean gates in
the circuit, extended to the set-valued gates to determine whether their set is
empty or not. We can easily compute~$\omega_\nu$ bottom-up 
from~$\nu$. Further, whenever $\nu$ is changed on a Boolean variable gate~$g$,
we can update $\omega_\nu$ by recomputing it bottom-up on~$\Delta(g)$. Formally:

\begin{lemmarep}
  \label{lem:omega}
  Given a hybrid circuit~$C$ of constant fan-in, given a valuation $\nu$ of~$C$, we can compute
  $\omega_\nu$ in linear time from~$\nu$ and~$C$. Further, for any $g \in
  C_\bvar$, letting $\nu'$ be the result of toggling the value of~$\nu$ on~$g$,
  we can update $\omega_\nu$ to $\omega_{\nu'}$ in time $O(\Delta(g))$.
\end{lemmarep}

\begin{proof}
  We explain how to compute $\omega_\nu$ bottom-up in linear time,
  in a way which is clearly correct by induction:
  \begin{itemize}
    \item For a Boolean variable gate $g$, we set $\omega_\nu(g) \colonequals \nu(g)$.
    \item For $\lor$-gates, $\land$-gates, and $\neg$-gates,
      we compute $\omega_\nu(g)$ from the
      $\omega_\nu$-value of the input gates with the Boolean operation indicated
      in the gate type.
    \item For a set-valued variable gate $g$, we set $\omega_\nu(g) \colonequals 1$.
    \item For a $\ctimes$-gate $g$, letting $g'$ and $g''$ be its two inputs, we
      set $\omega_\nu(g) \colonequals \omega_\nu(g') \land \omega_\nu(g'')$.
    \item For a $\atimes$-gate $g$ with two inputs $g'$ and $g''$, we
      set again $\omega_\nu(g) \colonequals \omega_\nu(g') \land \omega_\nu(g'')$.
    \item For a $\aplus$-gate $g$, letting $g_1, \ldots, g_n$ be its inputs, we
      set $\omega_\nu(g) \colonequals \bigvee_i \omega_\nu(g_i)$.
  \end{itemize}
  For updates, whenever $\nu$ is toggled on $g \in C_\bvar$, it is easy to see
  that, for any gate $g'$ of~$C$, if $\omega_\nu(g') \neq \omega_{\nu'}(g')$
  then $g' \in \Delta(g)$. Hence, we can update $\omega_\nu$ to $\omega_{\nu'}$
  within the prescribed time bound simply by taking $\omega_{\nu'}$ to be
  $\omega_\nu$ initially, and then recomputing $\omega_{\nu'}$
  on~$\Delta(g)$ according to the above scheme: this uses the fact that $C$ has
  constant fan-in.
\end{proof}

Hence, we can compute $\omega_\nu$ and $\delta$ in the
preprocessing and maintain them under updates.

\subparagraph*{Enumeration phase.}
We can use the shortcut function and partial evaluation to
enumerate the assignments in the set $\gsetv{\nu}{C}$ of our hybrid circuit~$C$.
Of course, if $\omega_\nu(g_0) = 0$ then we detect in constant time that there
is nothing to enumerate.
Otherwise, the enumeration scheme proceeds
essentially like in~\cite{amarilli2017circuit}; to achieve the right delay
bounds, it will need an additional index that we will present later.
We start by enumerating $\gsetv{\nu}{g_0}$, and describe what happens when we try
to enumerate $\gsetv{\nu}{g}$ for a set-valued gate~$g$; we will always ensure
that $\omega_\nu(g) = 1$. 
The base case is when $g$ is a set-valued variable, in which case the only
assignment to enumerate is~$\{g\}$. There are three induction cases:
$\atimes$-gates, $\ctimes$-gates, and $\aplus$-gates.

First, assume that $g$ is a $\atimes$-gate. As $C$ is homogenized, $g$ has two inputs $g_1$ and
$g_2$. Then we have $\gsetv{\nu}{g} = \gsetv{\nu}{g_1} \relprod \gsetv{\nu}{g_2}$. 
Hence, we can
simply enumerate $\gsetv{\nu}{g}$ as the lexicographic product
of~$\gsetv{\nu}{g_1}$
and $\gsetv{\nu}{g_2}$.
In particular, as 
$\omega_\nu(g) = 1$, we have $\omega_\nu(g_1) =
\omega_\nu(g_2) = 1$, so neither set is empty.
Formally, we have the
following lemma:

\begin{lemmarep}
  \label{lem:atimes}
  For any $\atimes$-gate $g$ with inputs $g_1$ and $g_2$, if we can enumerate
  $\gsetv{\nu}{g_1}$ and $\gsetv{\nu}{g_2}$ with delay and memory respectively
  $\theta_1$ and $\theta_2$, then we can
  enumerate $\gsetv{\nu}{g}$ with delay and memory $\theta_1 + \theta_2 + c$ for
  some constant~$c$.
\end{lemmarep}

\begin{proof}
  We enumerate all assignments for~$g_1$, which is non-empty because
  $\omega_\nu(g_1) = 1$; further, every assignment is
  non-empty because $C$ is homogenized.
  For each assignment $a$, we enumerate all assignments for~$g_2$, again a non-empty
  set of non-empty assignments, and for each assignment $a'$, we return the
  assignment $a
  \cup a'$, where the union is disjoint thanks to the determinism of~$g$.
  This satisfies the delay bound.
  The bound on memory usage is also satisfied,
  because we only need to remember the state in the enumeration on~$g_1$
  and~$g_2$ as well as a pointer to~$g$.
\end{proof}
Note that the constant~$c$ paid at the $\atimes$-gate is not a
problem to achieve linear delay and memory, because it is paid at most~$n-1$ times
when enumerating an assignment $A$ of size~$n$. Indeed, $C$ is
homogenized, so $A$ is always split non-trivially at each $\atimes$-gate, 
and $g$ is decomposable in~$\nu(C)$, so the
two sub-assignments never share any variable.

Second, assume that $g$ is a $\ctimes$-gate. As $\omega_\nu(g) = 1$, we clearly
have $\gsetv{\nu}{g} = \gsetv{\nu}{\delta(g)}$. Hence, we can simply follow the
pointer to~$\delta(g)$ and enumerate $\gsetv{\nu}{\delta(g)}$.
Intuitively, the cost of this operation can be covered by that of $g$,
because $\delta(g)$ can no longer be a $\ctimes$-gate.

\begin{lemmarep}
  \label{lem:ctimes}
  For any $\ctimes$-gate $g$, if we can enumerate
  $\gsetv{\nu}{\delta(g)}$ with delay and memory $\theta$, then we can
  enumerate $\gsetv{\nu}{g}$ with delay and memory $\theta+c$ for some
  constant~$c$.
\end{lemmarep}

\begin{proof}
  There is nothing to explain beyond what is given in the main text before the
  lemma statement.
\end{proof}

Third, assume that $g$ is a $\aplus$-gate $g$.
Naively, we can enumerate $\gsetv{\nu}{g}$
as the union of the $\gsetv{\nu}{g'}$ for the inputs $g'$ of~$g$ for which
$\omega_\nu(g') = 1$ (this union is disjoint thanks to determinism).
This is correct, but does not satisfy the delay bounds, because $g'$ may be another
$\aplus$-gate.
A more clever scheme is to 
to ``jump'' to the
$\atimes$-gates or set-valued variable gates on which $\gsetv{\nu}{g}$ depends.
Let us accordingly call \emph{exits}
the gates of these two types. The set $\gsetv{\nu}{g}$ can then be expressed
as a union of~$\gsetv{\nu}{g'}$ for the exits~$g'$
that have a directed path of $\aplus$-gates and $\ctimes$-gates to~$g$.
We introduce definitions to ``collapse'' these paths.

The first definition collapses paths of $\ctimes$-gates. There is a
\emph{$\ctimes$-path} from a set-valued gate $g'$ to a set-valued gate~$g \neq
g'$, written $g' \opath g$, if there is a
directed path $g' = g_1 \rightarrow \cdots \rightarrow g_n = g$ in~$C$ such
that $g_2, \ldots, g_{n-1}$ are all $\ctimes$-gates. In particular, a wire $(g',
g)$ between set-valued gates implies $g' \opath g$ (take $n=2$), and $\delta(g)
\opath g$ whenever $\delta(g) \neq g$.
When $g$ is a $\aplus$-gate, there are
two cases, depending on~$\nu$.
First, we may have $\omega_\nu(g_{n-1}) = 1$, 
and then 
$\omega_\nu(g') = 1$
and $\gsetv{\nu}{g'}$ contributes to
$\gsetv{\nu}{g}$: we call the path
\emph{live
under~$\nu$}.
Second, we may have
$\omega_\nu(g_{n-1}) = 0$, and then
$\gsetv{\nu}{g'}$ does not contribute
to~$\gsetv{\nu}{g}$ via this path.

The second definition collapses
paths of $\aplus$-gates. An \emph{$\aplus$-path} from a set-valued gate $g'$ to a
set-valued
gate~$g \neq g'$ is a sequence $g' = g_1 \opath \cdots \opath g_n = g$ in~$C$,
where $g_2, \ldots, g_{n-1}$ are all $\aplus$-gates and there is a $\ctimes$-path between any two consecutive gates.
The path is 
\emph{live under~$\nu$} if there is a \emph{live} $\ctimes$-path under~$\nu$ between
any two consecutive gates.

We now use these definitions to express $\gsetv{\nu}{g}$ as a function of the
set of \emph{exits under~$\nu$}
of~$g$ in~$C$, written $D^\nu_g$, which is the set of
exits $g'$ having a live $\aplus$-path
to~$g$ under~$\nu$ in~$C$:

\begin{lemmarep}
  \label{lem:aplussemantics}
  For any valuation $\nu$ and $\aplus$-gate $g$,
  we have $\gsetv{\nu}{g} = \bigcup_{g' \in D^\nu_g} \gsetv{\nu}{g'}$. Further, this union is
  disjoint and all its terms are nonempty.
\end{lemmarep}

\begin{proof}
  The first part of the result is easy to prove by bottom-up induction on the
  $\aplus$-gates. Specifically,
  for any $\aplus$-gate $g$, letting $g_1, \ldots, g_n$ be its
  inputs, we have by definition $\gsetv{\nu}{g} \colonequals \bigcup_i
  \gsetv{\nu}{g_i}$.
  For the $g_i$ which are $\ctimes$-gates, we can replace them by their one
  set-valued child provided that the $\omega_\nu$-image of their input under~$\nu$
  is~$1$, otherwise we can remove them from consideration. Repeating this
  process on~$\ctimes$-gates as long as possible, we eliminate some gates, and
  for those that remain, we reach a gate which is not a $\ctimes$-gate
  (specifically, the $\delta$-image of the corresponding~$g_i$).
  By this reasoning, it is clear that $\gsetv{\nu}{g}$ is the union of the
  $\gsetv{\nu}{g_i}$ for the set $S_1$ of inputs~$g_i$ of~$g$ which are not $\ctimes$-gates, unioned to the
  union of the~$\gsetv{\nu}{g'}$ for
  the set $S_2$
  of gates~$g'$ reached by going down live paths of $\ctimes$-gates as we
  explained. Now, as $D^\nu_g$
  consists of the gates of~$S_1$, plus the sets
  $D^\nu_{g''}$ for the gates $g''$ of the set~$S_2$, we can
  conclude by induction that the claim made in the first sentence of the lemma holds
  for~$g$, concluding the proof of the first part.

  For the second part of the claim, the fact that the union is disjoint is
  thanks to determinism: assuming by contradiction that there is an assignment $a$
  such that $a \in \gsetv{\nu}{g'}$ and $a \in \gsetv{\nu}{g''}$ for two different gates $g' \neq g''$ of~$D^\nu_g$,
  consider a live $\aplus$-path $\pi' : g' = g'_1 \opath \cdots \opath g'_n = g$ from
  $g'$ to~$g$, and a live $\aplus$-path $\pi'': g'' = g''_1 \opath \cdots \opath g''_m
  = g$ from~$g''$ to~$g$. Let $g''' = g'_i = g''_j$ be the first gate where
  these two paths join; we have $i>1$ and $j>1$ because $g' \neq g''$. Hence, we
  know that $g'''$ is a $\aplus$-gate such that $g'_{i-1}$ and $g''_{j-1}$ both
  have a live $\ctimes$-path to~$g'''$; further $g'_{i-1} \neq g''_{j-1}$
  because $g'''$ is the first gate where the paths join. The paths $\pi'$ and
  $\pi''$
  clearly witness that $a \in \gsetv{\nu}{g'_{i-1}}$ and $a \in
  \gsetv{\nu}{g''_{j-1}}$. Now, let $h'$ be the last gate of a witnessing live
  $\ctimes$-path $\rho'$ from $g'_{i-1}$ to $g'''$, and $h''$ be the last gate of a
  witnessing live $\ctimes$-path $\rho''$ from $g''_{j-1}$ to~$g'''$. We have
  $h' \neq h''$ (otherwise we would have also $g'_{i-1} = \delta(h') =
  \delta(h'') = g''_{j-1}$, contradicting our earlier claim). Again $\rho'$ and
  $\rho''$ witness that $a \in \gsetv{\nu}{h'}$ and $a \in \gsetv{\nu}{h''}$.
  But $h'$ and $h''$ are two different inputs of the $\aplus$-gate $g'''$, so we
  have witnessed a violation of determinism, a contradiction. Hence, indeed, the
  union is disjoint.

  Last, the fact that none of the terms is empty is because the definition of
  live $\aplus$-paths enforces that $\omega_\nu(g') = 1$ for all $g' \in
  D^\nu_g$, so we conclude by definition of~$\omega_\nu$.
\end{proof}

Hence, we can enumerate $\gsetv{\mu}{g}$ for a $\aplus$-gate $g$ by enumerating
$D^\nu_g$ and the set $\gsetv{\nu}{g'}$ for each~$g'$ in~$D^\nu_g$. Note that
$g'$ is an exit, i.e., a variable 
or a $\atimes$-gate; so we make progress.

\begin{lemmarep}
  \label{lem:aplus}
  For any $\aplus$-gate $g$, if we can enumerate $D^\nu_g$ with delay and
  memory $c$, and can enumerate $\gsetv{\nu}{g'}$ for every $g' \in D^\nu_g$ with 
  delay and memory $\theta$, then we can enumerate $\gsetv{\nu}{g}$ with delay and
  memory $\theta + c + c'$ for some constant~$c'$.
\end{lemmarep}

\begin{proof}
  This follows immediately from Lemma~\ref{lem:aplussemantics} and the
  explanations given in the main text before the lemma statement.
\end{proof}

We have described our enumeration scheme in Lemmas~\ref{lem:atimes}, \ref{lem:ctimes}, and
\ref{lem:aplus}. The only missing piece is to enumerate, for each
$\aplus$-gate~$g$, the set
$D^\nu_g$ of exits
under~$\nu$ of~$g$, with \emph{constant} delay and memory.
To do so, we will need additional preprocessing.
We will rely on upwards-determinism, and extend the tree-based index
of~\cite{amarilli2017circuit_extended} to support updates.
We first
present an additional structure, called the \emph{switchboard}, that we compute
in the preprocessing; and we explain in the next section an indexing scheme
that we perform on this structure.

\subparagraph*{Switchboard.}
Our third index component in the preprocessing is called the
\emph{switchboard}. It consists of a directed graph $B = (V, E)$ called the
\emph{panel}, which does not depend on~$\nu$ (so it does not need to be
updated), and a valuation $\beta_\nu:E\to\{0, 1\}$ called the \emph{wiring}.
The \emph{panel} $B = (V, E)$ is defined as follows: $V$ consists of all $\aplus$-gates,
$\atimes$-gates, and $\avar$-gates, and $E \subseteq V \times V$ contains the
edge $(\delta(g'), g)$ for each wire $(g', g)$ of~$C$ such that $g$
is a $\aplus$-gate. This implies that the maximal fan-in of~$B$ is no greater than
that of~$C$, and it implies that $B$ is a DAG.
The \emph{wiring} $\beta_\nu$ maps every edge $(g', g)$ of~$B$ to~$1$ if there is 
a $\ctimes$-path from~$g'$ to~$g$ in~$C$ which is live under~$\nu$, and~$0$
otherwise.
We can use $\omega_\nu$ to compute the switchboard, and to update it
in time $O(\Delta(C))$ whenever $\nu$ is updated by toggling a gate
of~$C_\bvar$. Formally:

\begin{lemmarep}
  \label{lem:switchboard}
  The switchboard can be computed in linear time given~$C$
  and~$\nu$, and we can update it in time $O(\Delta(C))$ when toggling any gate
  in~$\nu$.
\end{lemmarep}

\begin{proof}
  We can compute $B$ in linear time during the preprocessing by going over all
  edges of~$C$, and using the characterization: for each wire $(g', g)$ of~$C$,
  if $g$ is a $\aplus$-gate, we add to~$E$ the edge $(\delta(g'), g)$, using
  $\delta$ which has already been computed.

  We can compute $\beta_\nu$ in linear time during the preprocessing, from
  $\omega_\nu$ (which has already been computed): initialize $\beta_\nu$ by
  mapping all edges of~$B$
  to~$0$, and for each wire $(g', g)$ of~$C$ such that $g$ is a $\aplus$-gate,
  set $\beta_\nu((\delta(g'), g)) = 1$ if $\omega_\nu(g') = 1$.

  For the claim on updating the wiring, let us define
  the \emph{dependent gates} $\Delta_E(g')$ of a gate~$g'$ in the switchboard~$B$
like we did for circuits, i.e., the set of gates $g$ such that there is a
  directed path from~$g'$ to~$g$ in the switchboard~$B$.
  Observe now that, by construction, for any
  gate $g$ of~$C$, we have $\Delta_E(g) \subseteq \Delta(g)$. Now, 
  when we update $\nu$ by toggling the value of $g \in C_{\bvar}$, then we can
  update $\beta_\nu$ in $O(\card{\Delta(g)})$. Indeed, we know that we can
  update $\omega_\nu$ in this time, and that it only changes on gates
  in~$\Delta(g)$. Hence, we can map to~$0$ all edges $(g'', g')$ of~$E$ such
  that $g' \in \Delta(g)$, and recompute $\beta_\nu$ on these edges. As $B$ has
  constant fan-in, the number of such edges is in $O(\card{\Delta_E(g)})$, which is
  $O(\card{\Delta(g)})$, achieving the bound.
\end{proof}

We now explain how we use the switchboard to enumerate, 
given a $\aplus$-gate $g$,
the set $D^\nu_g$ of the exits $g'$ having a \emph{live} $\aplus$-path to~$g$ under~$\nu$. In terms
of the switchboard, we must enumerate the exits $g'$ that
have a path to~$g$ in~$B$ whose
edges are all mapped to~$1$ by~$\beta_\nu$. Hence, we must solve the
following enumeration task on the switchboard: letting $\beta_\nu(B)$ be the DAG
of edges of~$B$ mapped to~$1$ by~$\beta_\nu$, we are given a gate $g$ of~$B$,
and we must enumerate all exit gates $g'$ of~$B$ (i.e., the $\atimes$-gates or
$\avar$-gates) that have a directed path to~$g$ in~$\beta_\nu(B)$. Further, we must
be able to handle updates on~$\beta_\nu(B)$, as given by updates on~$\nu$.
Fortunately, thanks to upwards-determinism, this problem is easier than it
looks:

\begin{claimrep}
  \label{clm:switchboard_forest}
  For any valuation $\nu$ of the hybrid circuit $C$, the DAG $\beta_\nu(B)$
  is a forest.
\end{claimrep}

\begin{proof}
  The claim can be equivalently rephrased as follows: there
  is no gate $g \in V$ such that $\beta_\nu((g, g')) = \beta_\nu((g, g'')) = 1$
  for two different gates $g' \neq g''$.
  (Pay attention to the fact that the edges of the forest are oriented upwards
  rather than downwards, following the direction of the wires in circuits.)

  To show this, 
  let us assume to the contrary that there is a valuation $\nu$ such that there
  are gates $g$ and $g' \neq g''$ with $\beta_\nu((g, g')) = \beta_\nu((g, g'')) = 1$,
  and let us conclude a violation of upwards-determinism. First, these
  $\beta_\nu$-values imply in particular that $\omega_\nu(g) = 1$. Now, consider two
  witnessing live $\ctimes$-paths $g = g_1' \rightarrow \cdots \rightarrow g_n' = g'$ and 
  $g = g_1'' \rightarrow \cdots \rightarrow g_m'' = g''$ where
  $g_2', \ldots, g_{n-1}'$ and $g_2'', \ldots, g_{m-1}''$ are $\ctimes$-gates
  whose $\omega_\nu$-image is~$1$. Let $g''' = g'_i = g''_j$ be the last common
  gate of these two paths; as $g' \neq g''$, we have $i < n$ and $j < m$.
  Consider the wires $(g''', g'_{i+1})$ and $(g''', g''_{j+1})$. The gate
  $g'_{i+1}$ is either $g'$, in which case it is a $\aplus$-gate and the wire is
  pure, or it is a $\ctimes$-gate whose $\omega_\nu$-image is~$1$, i.e., its
  second input evaluates to~$1$ under~$\nu$, so the wire is pure.
  The same is true of~$g''_{j+1}$.
  Hence, these two wires witness a violation of the upwards-determinism
  condition on~$g'''$ in~$\nu(C)$. This is a contradiction, which concludes the
  proof.
\end{proof}

\begin{example}
  \label{exa:switchboard}
  Figure~\subref{fig:switchboard} describes the
  switchboard for the hybrid circuit $C$ of Figure~\subref{fig:hybrid}. The
  edges of the switchboard correspond to $\ctimes$-paths. The
  switchboard itself is not a forest; however, for every valuation of~$C$,
  the $\ctimes$-paths that are live must always form a forest.
\end{example}

Thus, what we need is a constant-delay reachability index on forests that can be
updated efficiently when adding and removing edges to the forest. This is the
focus of the next section.

\section{Reachability Indexing under Updates}
\label{sec:reachability}
In this section, we present our indexing scheme for reachability on forests
under updates. The construction in this section is independent from what
precedes. For convenience, we will orient the edges of the forest downwards,
i.e., the reverse of the previous section (so $g$ is the parent of~$g'$ in the
forest if there is an edge from~$g'$ to~$g$ in the switchboard). We first define
the problem and state the enumeration result, and then sketch the proof.

\begin{toappendix}
  \subsection{Proof of Theorem~\ref{thm:forest}}
\end{toappendix}

\subparagraph*{Definitions and main result.}
A \emph{reachability forest} $F = (V, E, X)$ is a directed graph $(V, E)$ where
$V$ is the \emph{vertex set}, $E \subseteq V \times V$ are the \emph{edges}, and
$X\subseteq V$ is a subset of vertices called \emph{exits}. When $(v, v') \in
E$, we call $v$ a \emph{parent} of~$v'$, and $v'$ a \emph{child} of~$v$.
We impose three 
requirements on~$F$: 
(i.) the graph $(V, E)$ is a forest, i.e.,
each vertex of~$V$ has at most one parent;
(ii.) there is a constant \emph{degree bound} $c \in \NN$ such that every vertex
has at most~$c$ children;
(iii.) every exit $v\in X$ is a
\emph{leaf}, i.e., a vertex with no children.
We will call \emph{trees} the connected components of~$F$.
For convenience, we assume that~$F$ is \emph{ordered}, i.e., there is some
total order $<$ on the children of every node.

Given a reachability forest $F = (V, E, X)$ and a vertex $v \in V$,
we write
$\reach(v)$ for the set of exits reachable in $F$ from~$v$, i.e.,
the vertices of~$X$ to which~$v$ has a directed path. These
are the sets that we wish to enumerate efficiently, allowing two
kinds of \emph{updates} on the edges~$E$ of~$F$.
First, a \emph{delete
operation} is written $-E'$ for a set~$E' \subseteq E$, and 
$F = (V, E, X)$ is updated to $F-E' \colonequals (V, E \setminus E', X)$; it is
still a reachability forest. Second, an
\emph{insert operation} is written $+E'$ for some $E' \subseteq V \times V$,
and we require that the 
update result $F + E' \colonequals (V, E \cup E', X)$
still satisfies the three requirements above (with the same degree bound).
In terms of the order~$<$ on children, when we remove edges, we take
the restriction of~$<$ in the expected way, and when we insert edges, we add
each new child at an arbitrary position in~$<$.
We then introduce \emph{ancestry} to measure the impact of updates (analogously
to dependency size):
the \emph{ancestry} $\anc_F(v)$ of $v \in V$ is the set of vertices of~$F$
that have a directed path to~$v$, and
the \emph{ancestry} $\anc_F(E')$ for $E' \subseteq V \times V$ is $\bigcup_{(v,
w) \in E'} \anc_F(v)$.
We then have:

\begin{toappendix}
  In this appendix, we prove Theorem~\ref{thm:forest}.
  Recall the result
  statement:
\end{toappendix}

\begin{theoremrep}
  \label{thm:forest}
  Given a reachability forest~$F$,
  there is an enumeration algorithm with linear-time preprocessing such
  that: (i.) given any $v \in V$, we can enumerate $\reach(v)$ with
  constant delay and memory; (ii.) given an update $\pm E'$, we can
  apply it (replacing $F$ by $F \pm E'$ and updating the index) with
  update time in $O(\anc_F(E'))$.
\end{theoremrep}

\begin{toappendix}
  To show this
  result, we only need to argue that we can compute and update the pointers in
  our index (illustrated on an example in Figure~\ref{fig:forest}).
  Indeed, as we have explained in the main text, when we have these pointers, we
  can use them to perform enumeration with constant delay and memory.

  \tikzstyle{node}=[draw,circle]
\tikzstyle{exit}=[draw,rectangle]
\tikzstyle{edge}=[ultra thick]
\newcommand{\dside}[3][red]{
    \path[red,dashed,#1] (#2) edge (#3);
}
\newcommand{\ddown}[4][blue]{
  \path[blue,#1] (#2) edge[bend right=20] (#3);
  \path[blue,#1] (#2) edge[bend left=20] (#4);
}
\newcommand{\ddownsl}[1]{
  \path[blue] (#1) edge[out=360,in=300,looseness=5] (#1);
  \path[blue] (#1) edge[out=180,in=240,looseness=5] (#1);
}

\begin{figure}
  \centering
  \begin{tikzpicture}[xscale=1.9]
    \node[node] (a) at (0, 0) {};
    \node[node] (a1) at (-1, -1) {};
    \node[node] (a2) at (1, -1) {};
    \node[exit] (a21) at (.5, -2) {};
    \node[node] (a22) at (1.5, -2) {};
    \node[node] (a11) at (-1.5, -2) {};
    \node[exit] (a12) at (-.5, -2) {};
    \node[node] (a111) at (-2, -3) {};
    \node[exit] (a112) at (-1, -3) {};
    \draw[edge] (a) -- (a1);
    \draw[edge] (a) -- (a2);
    \draw[edge] (a2) -- (a21);
    \draw[edge] (a2) -- (a22);
    \draw[edge] (a1) -- (a11);
    \draw[edge] (a1) -- (a12);
    \draw[edge] (a11) -- (a111);
    \draw[edge] (a11) -- (a112);
    \dside{a112}{a12}
    \dside{a12}{a21}
    \ddown{a}{a112}{a21}
    \ddown{a1}{a112}{a12}
    \ddown{a11}{a112}{a112}
    \ddown{a2}{a21}{a21}
    \ddownsl{a112}
    \ddownsl{a12}
    \ddownsl{a21}

    \node[node] (b) at (2.5, 0) {};
    \node[exit] (b1) at (2, -1) {};
    \node[exit] (b2) at (3, -1) {};
    \draw[edge] (b) -- (b1);
    \draw[edge] (b) -- (b2);
    \dside{b1}{b2}
    \ddownsl{b1}
    \ddownsl{b2}
    \ddown{b}{b1}{b2}

    \node[exit] (c) at (5, 0) {};
    \ddownsl{c}

    \node[node] (d) at (4, 0) {};
    \node[node] (d1) at (4, -1) {};
    \node[node] (d11) at (4, -2) {};
    \node[node] (d111) at (3.5, -3) {};
    \node[node] (d112) at (4.5, -3) {};
    \draw[edge] (d) -- (d1);
    \draw[edge] (d1) -- (d11);
    \draw[edge] (d11) -- (d111);
    \draw[edge] (d11) -- (d112);
  \end{tikzpicture}

  \caption{Illustration of the index structure for Theorem~\ref{thm:forest} on
  an example reachability forest (drawn with thick edges).
  Exits are drawn as squares, other vertices are drawn as circles.
  The next pointers are drawn in straight dashed red lines (from left to right), and the first and
  last pointers of each node are drawn as curved solid blue lines (from top to
  bottom), respectively at the left and right of
  the node. Pointers that are $\snull$ are not drawn.}
  \label{fig:forest}
\end{figure}

  Hence, the only thing to show is the following result:
\end{toappendix}

Note how we can insert (or delete) many edges at the same time, paying
only once the price $\anc_F(E')$: this point is used in the proof of
Theorem~\ref{thm:circuits}
to bound the total
cost of each update on the circuit.
We sketch
the proof of Theorem~\ref{thm:forest} in the rest of this section.

\subparagraph*{Construction for Theorem~\ref{thm:forest}.}
Our index structure 
follows the one used to prove Proposition~F.4
of~\cite{amarilli2017circuit_extended}: it
maps every $v \in V$ to
a pointer $\first_F(v)$ and a pointer $\last_F(v)$, called the \emph{$\first$} 
and \emph{$\last$ pointer}; and
maps every exit $v \in X$ to a 
pointer $\pnext_F(v)$ called the \emph{$\pnext$ pointer}.
These pointers are defined using the order
$<'$
given by a preorder traversal of~$F$ following~$<$.
Specifically, $\first_F(v)$ is the first exit $v' \in \reach_F(v)$ according
to~$<'$, and $\last_F(v)$ is the last such exit; if $\reach_F(v) = \emptyset$ then both
pointers are~$\snull$. Now, $\pnext_F(v)$ for~$v \in X$ is the exit $v' \in
X$ in the tree of~$v$ which is the successor of~$v$ according to~$<'$; if $v$ is the last exit
of its tree, then $\pnext_F(v)$ is~$\snull$.
If we know these pointers, we can enumerate $\reach_F(v)$ for any
$v \in V$ with constant delay and memory as
in~\cite{amarilli2017circuit_extended}: if $\first_F(v)$ is~$\snull$ then there
is nothing to enumerate, otherwise start at $v_- \colonequals \first_F(v)$,
memorize $v_+ \colonequals \last_F(v)$, and enumerate the reachable
exits following the $\pnext$ pointers from~$v_-$ until reaching~$v_+$.
Hence, to conclude the proof of Theorem~\ref{thm:forest},
it suffices to compute and update these pointers efficiently:

\begin{lemmarep}
  \label{lem:forestclaim}
  Given a reachability forest~$F$,
  we can compute the $\first$, $\last$, and $\pnext$ pointers of all vertices
  in time~$O(\card{F})$.
  Further, for any update $\pm E'$, we can
  apply it and update the pointers in time $O(\anc_F(E'))$.
\end{lemmarep}

\begin{proofsketch}
  The $\first$ and $\last$ pointers are computed bottom-up in linear time:
  for a leaf~$v$, they either point to~$v$ if $v\in X$ or to~$\snull$ otherwise; 
  for an
  internal vertex $v$, we set $\first_F(v)$ as $\first_F(v')$ for the smallest
  child~$v'$ of~$v$ in the order~$<'$ with a non-$\snull$ $\first$ pointer (or $\snull$ if all
  $\first$
  pointers of children are~$\snull$), and we set $\last_F(v)$ analogously, using
  the $\last$ pointer of the largest child of~$v$ in the order~$<'$ for which the $\last$ pointer is non-$\snull$. Further, given
  an update $\pm E'$, the $\first$ and $\last$ pointers need only to be updated
  in~$\anc_F(E')$, and we can recompute them there with the same bottom-up
  scheme.

  The $\pnext$ pointers are also computed bottom-up in linear
  time: at each internal vertex $v$, we go over its children and
  stitch together the sequences of $\pnext$ pointers of their subtrees. 
  Specifically, when $\last_F(v_1)$ is not~$\snull$ for a child~$v_1$,
  we find the next child $v_2$ for which
  $\first_F(v_2)$ is not~$\snull$, and set $\pnext_F(\last_F(v_1))
  \colonequals \first_F(v_2)$. Again, for an update $\pm E'$,
  we recompute the
  $\pnext$
  pointers by processing $\anc_F(E')$ bottom-up in a similar fashion.
\end{proofsketch}

\begin{toappendix}
  Before we show the result, we make a simple observation on the complexity of
  updates. Remember that $\anc_F(E')$ refers to the ancestry of the parent edges
  of the vertices of~$E'$ in~$F$, i.e., \emph{before} the update is performed.
  We will sometimes want to process $\anc_{F \pm E'}(E')$, i.e., the ancestry
  of~$E'$ in the forest \emph{after} the update. However, the distinction
  between the two is inessential, because of the following result:

  \begin{claim}
    \label{clm:ancestry}
    For any reachability forest $F$ and update $\pm E'$, we have 
    $\anc_F(E') = \anc_{F \pm E'}(E')$.
  \end{claim}

  \begin{proof}
    It suffices to show the claim for deletions. Indeed, for any forest $F$ and
    insertion $+ E'$, letting $F' \colonequals F + E'$, we have $F = F' - E'$,
    so we can simply apply the claim to~$F'$ and to the deletion $-E'$.

    Now, for deletions, we know that $\anc_F(E') \supseteq \anc_{F - E'}(E')$,
    because obviously $\anc_F(v) \supseteq \anc_{F - E'}(v)$ for any vertex
    $v\in V$. Conversely, let us show that $\anc_F(E') \subseteq \anc_{F -
    E'}(E')$ by showing that, for each $(v, w) \in E'$, we have $\anc_F(v)
    \subseteq \anc_{F - E'}(E')$. Consider the chain of ancestors of~$v$ in~$F$,
    and the edges between them (not including $(v, w)$): either none of these
    edges is in~$E'$, in which case we have $\anc_F(v) = \anc_{F-E'}(E')$, or
    some edges are. In this case, considering all edges $E''$ on this path that
    are in~$E'$,
    it is easy to see that the union of $\anc_{F- E'}(e)$ for $e \in E'' \cup
    \{(v, w)\}$ is exactly $\anc_F(v)$. Hence, indeed we have $\anc_F(v)
    \subseteq \anc_{F - E'}(E')$, which establishes the reverse inclusion and
    concludes the proof.
  \end{proof}
  
  We are now ready to prove Lemma~\ref{lem:forestclaim} (see also an
  illustration in Figure~\ref{fig:update}).

  \begin{figure}
  \hfill\begin{tikzpicture}[xscale=1.25, baseline={(current bounding box.center)}]
    \path[use as bounding box] (-2, 0) rectangle (2, -3.3);
    \node[node] (a) at (.5, 0) {};
    \node[node] (a1) at (-1, -1) {};
    \node[node] (a2) at (2, -1) {};
    \node[node] (a11) at (-2, -2) {};
    \node[node] (a12) at (0, -2) {};
    \node[node] (a21) at (2, -2) {};
    \node[exit] (a111) at (-2, -3) {};
    \node[exit] (a121) at (-.75, -3) {};
    \node[exit] (a122) at (.75, -3) {};
    \node[exit] (a211) at (2, -3) {};

    \draw[edge] (a) -- (a1);
    \draw[edge] (a1) -- (a11);
    \draw[edge] (a11) -- (a111);
    \draw[edge] (a12) -- (a121);
    \draw[edge] (a12) -- (a122);
    \draw[edge] (a) -- (a2);
    \draw[edge] (a2) -- (a21);
    \draw[edge] (a21) -- (a211);

    \dside{a121}{a122}
    \path[red,dashed] (a111) edge[bend right=20] (a211);

    \ddown{a11}{a111}{a111}
    \ddown{a12}{a121}{a122}
    \ddown{a2}{a211}{a211}
    \ddown{a21}{a211}{a211}
    \path[blue] (a) edge[bend right=90] (a111);
    \path[blue] (a) edge[bend left=90] (a211);
    \path[blue] (a1) edge[bend right=10] (a111);
    \path[blue] (a1) edge[bend left=30] (a111);

    \ddownsl{a111}
    \ddownsl{a121}
    \ddownsl{a122}
    \ddownsl{a211}
  \end{tikzpicture}
  \hfill
  \hfill
  {\huge $\Rightarrow$}
  \hfill
  \hfill
  \begin{tikzpicture}[xscale=1.25, baseline={(current bounding box.center)}]
    \path[use as bounding box] (-2, 0) rectangle (2, -3.3);
    \node[node] (a) at (.5, 0) {};
    \node[node] (a1) at (-1, -1) {};
    \node[node] (a2) at (2, -1) {};
    \node[node] (a11) at (-2, -2) {};
    \node[node] (a12) at (0, -2) {};
    \node[node] (a21) at (2, -2) {};
    \node[exit] (a111) at (-2, -3) {};
    \node[exit] (a121) at (-.75, -3) {};
    \node[exit] (a122) at (.75, -3) {};
    \node[exit] (a211) at (2, -3) {};

    \draw[edge] (a) -- (a1);
    \draw[edge] (a1) -- (a11);
    \draw[edge] (a11) -- (a111);
    \draw[edge] (a12) -- (a121);
    \draw[edge] (a12) -- (a122);
    \draw[edge] (a1) -- (a12);
    \draw[edge] (a) -- (a2);
    \draw[edge] (a2) -- (a21);
    \draw[edge] (a21) -- (a211);

    \dside{a111}{a121}
    \dside{a122}{a211}
    \dside{a121}{a122}

    \ddown{a11}{a111}{a111}
    \ddown{a12}{a121}{a122}
    \ddown{a2}{a211}{a211}
    \ddown{a21}{a211}{a211}
    \path[blue] (a) edge[bend right=90] (a111);
    \path[blue] (a) edge[bend left=90] (a211);
    \path[blue] (a1) edge[bend right=10] (a111);
    \path[blue] (a1) edge[bend left=30] (a122);

    \ddownsl{a111}
    \ddownsl{a121}
    \ddownsl{a122}
    \ddownsl{a211}
  \end{tikzpicture}
  \hfill
  \null

  \caption{Example for updating the index of Theorem~\ref{thm:forest} when
  inserting an edge (from the left forest to the right forest). We follow the same
  drawing conventions as in Figure~\ref{fig:forest}.}
  \label{fig:update}
\end{figure}

  \begin{proof}[Proof of Lemma~\ref{lem:forestclaim}]
    Let $F' \colonequals F \pm E'$.
    We first show the result for the $\first$ and $\last$ pointers. For the initial
    computation, we use the scheme described in the proof sketch; it clearly runs
    in linear time (it examines every edge of~$F$ once), and it is correct by a
    straightforward induction.

    To update the $\first$ and $\last$ pointers, we observe that, for every vertex $v$ \emph{not} in
    $\anc_F(E')$, the pointers do not need to be changed: this is clear because,
    for every such $v$, the subtree in~$F'$ rooted at~$v$ is exactly the
    same as in~$F$. Hence, outside of~$\anc_F(E')$, the $\first$ and $\last$ pointers
    are still correct, so it suffices to update
    the pointers in~$\anc_F(E')$. We do this by the same bottom-up scheme as for
    the initial computation. Specifically, for deletions, we process $\anc_F(E')$
    (i.e., the ancestry in the original~$F$, \emph{before} the update), but we
    perform the computation at each node based on its children \emph{after} the
    update (i.e., ignoring children whose parent edge has just been deleted).
    For insertions, we process $\anc_{F'}(E')$, i.e., the ancestry
    in~$F'$ \emph{after} the update, and perform the computation at each node
    based its the children after the update (i.e., after all insertions have been
      performed): this gives the right complexity thanks to
      Claim~\ref{clm:ancestry}.
    The correctness of this update scheme is again shown by
    induction, using the additional base case that consists of the vertices
    outside of~$\anc_F(E')$, which are correct as we explained. The complexity is
    in $O(\card{\anc_F(E')})$, because we examine edges in the set~$\anc_F(E')$
    and child edges of vertices of this set, so at most $(c+1) \card{\anc_F(E')}$
    where~$c$ is the constant degree bound. This concludes the proof for the
    $\first$
    and $\last$ pointers.

    \medskip

    For the $\pnext$ pointers, we give a more precise description of the scheme presented in the
    proof sketch. We process~$F$ bottom-up and ensure that, whenever we are done
    processing a vertex $v\in F$, then the $\pnext$ pointers within the subtree
    rooted at~$v$ are correct; but we do not specify anything about the $\pnext$
    pointer of the last exit in this subtree. Initially, we set all $\pnext$
    pointers to~$\snull$, which is correct as a base case for the leaves. Now,
    to process $v\in F$ with children $v_1 < \cdots < v_n$, assuming by
    induction that the $\pnext$ pointers within each subtree rooted at~$v_1$ are
    correct, we simply need to go over the~$v_i$ in order, maintaining a
    \emph{current last exit} $v'$ which denotes the last exit among all $v_i$
    seen so far, whose $\pnext$ pointer is currently~$\snull$. The current last exit $v'$ is
    initially $\snull$. When we look at~$v_i$, if $\first_F(v) = \snull$, we do
    nothing. Otherwise, if $v' = \snull$, then we replace $v'$ by $\last_F(v)$
    and do nothing more.
    Otherwise, if $v'$ is not null, then letting $v'' \colonequals
    \first_F(v_i)$, we set $\pnext_F(v') \colonequals v''$ and we set $v'
    \colonequals v''$ as our new current last exit. It is clear that this
    process satisfies our invariant.

    Now, when we process the root $v$ of a tree using this scheme,
    in the case where the last reachable exit
    $v' \colonequals \last_F(v)$ is not~$\snull$ at the end of the process,
    our invariant does not guarantee anything about $\pnext_F(v')$;
    but we can simply ensure that the $\pnext$ pointers
    in that tree are correct (including the $\pnext$ pointer of the last exit) by
    setting $\pnext_F(v') \colonequals \snull$. The overall scheme clearly runs in
    linear time for the initial computation, and it is inductively correct.
      
    We conclude by explaining the update scheme for the $\pnext$ pointers. 
    Note that, this time, it is no longer the case that the $\pnext$ pointers to be
    updated are all in $\anc_F(E')$: see Figure~\ref{fig:update} for an example.
    However, intuitively, the vertices outside of~$\anc_F(E')$ whose
    $\pnext$-pointers need to be
    updated are all reachable as the value of a $\last$ pointer for a vertex of~$\anc_F(E')$, so
    we can fix all pointers by re-running our bottom-up computation scheme
    on~$\anc_F(E')$.
    Initially, we keep $\pnext_{F'}(v) \colonequals \pnext_F(v)$, which requires
    no modifications on the index. As in the preprocessing, we will ensure as an
    invariant when performing the update that, when we are done processing a vertex~$v$, the $\pnext$ pointers
    in the subtree rooted at~$v$ in~$F'$ are all correct (specifically, for
    every reachable exit of~$X$ in~$F'$, the $\pnext$ pointer correctly points to
    the next exit if it exists); but again we do not specify anything about the
    $\pnext$ pointer of the last exit of this subtree.
    Note that the invariant is already satisfied for
    all vertices not in $\anc_F(\pm E')$: their reachable subtree is unchanged
    between~$F$ and~$F'$, so all $\pnext$ pointers within the subtree are still
    correct.

    We now process $\anc_F(\pm E')$ bottom-up, relying on the above observation
    for the base case, and doing the inductive case exactly as in the
    preprocessing algorithm above. Like in the update scheme for $\first$ and
    $\last$
    pointers, we process $F'$ in the case
    of insertions (using Claim~\ref{clm:ancestry} to ensure that the size bound
    is correct), and process $F$ in the case of deletions (but at each node we
    do not take 
    into account the children corresponding to edges of~$F$ that are deleted
    in~$F'$). This processing allows us to ensure that all nodes in~$F'$ satisfy
    the invariant. Now, as before, once we
    have processed a vertex which is the root of a tree in $F'$,
    then we set the $\pnext$ pointer of the last reachable
    exit of the root of this tree to~$\snull$: this ensures that, in addition to the
    invariant, all $\pnext$ pointers in its tree
    are correct (including the last one).

    At the end of this processing, the invariant is ensured on all trees, and
    further we know that the last exit of each tree correctly has~$\snull$ as
    its $\pnext$ pointer. Hence, we have correctly recomputed the $\pnext$ pointers in
    the prescribed time bound. Hence, we have explained how to handle updates
    for the $\pnext$ pointers, which concludes the proof of
    Lemma~\ref{lem:forestclaim}.
  \end{proof}
\end{toappendix}

\begin{toappendix}
  \subsection{Putting Everything Together}
\label{apx:together}

In this appendix section, we recap the proof of our main results. We first prove
Theorem~\ref{thm:circuits}:

\begin{proof}[Proof of Theorem~\ref{thm:circuits}]
  We apply the scheme of Section~\ref{sec:enumeration}. Given $C$ and~$\nu$,
  we compute the shortcut function $\delta$,
  the partial evaluation $\omega_\nu$,
  and the switchboard composed of the panel $B = (V, E)$ 
  and its wiring $\beta_\nu$.
  Further, we compute
  the index structure of Theorem~\ref{thm:forest} on the DAG
  $\beta_\nu(B)$ of the edges of~$B$ mapped to~$1$ by~$\nu$ (choosing any
  arbitrary order on the children of each vertex),
  which is a forest
  by Claim~\ref{clm:switchboard_forest}.
  Note that the \emph{exits} (non-$\aplus$-gates) in the forest are the exits in
  the sense of Section~\ref{sec:enumeration} (i.e., the $\atimes$-gates and
  $\avar$-gates in the panel), which are leaves by definition of the panel. Also
  note that, in the reachability forest, 
  all nodes have degree no greater than the maximal fan-in of~$C$ (so we can use it
  as degree bound~$c$).
  Keep in mind that the direction is
  reversed between $\beta_\nu(B)$ as defined in Section~\ref{sec:enumeration},
  and the reachability forest as studied in Section~\ref{sec:reachability}.
  Indeed,
  in the circuit, we want to enumerate the reachable exits of
  a $\aplus$-gate $g$ in the sense of having a path (specifically, a live
  $\aplus$-path) to~$g$, whereas in the reachability forest, we enumerate the
  exits to which~$g$ has a directed path. However, this is fine because
  upwards-determinism guarantees in Claim~\ref{clm:switchboard_forest} that
  $\beta_\nu(B)$ is a forest where the edges are oriented upwards (see the proof
  for details), so reversing
  the edges gives a forest in the sense of Section~\ref{sec:reachability}.
  This concludes the description of our preprocessing scheme, which runs in
  linear time.

  \medskip

  To enumerate the assignments of the circuit, we use the scheme described by
  Lemmas~\ref{lem:atimes}, \ref{lem:ctimes}, and \ref{lem:aplus}, as well as the
  explanations around them in the main text. To enumerate the set $D^\nu_g$ of
  the reachable exits of~$g$ for Lemma~\ref{lem:aplus},
  i.e., the exits having a directed path to~$g$ in~$\beta_\nu(B)$, we enumerate
  the set $\reach_F(g)$ in the reachability forest, which is precisely what can
  be done with the index of Theorem~\ref{thm:forest}. We summarize why the
  enumeration is in delay and memory linear in each produced assignment:
  \begin{itemize}
    \item Whenever we reach a $\ctimes$-gate, we pay constant delay and memory and reach
      a gate which is not a $\ctimes$-gate;
    \item Whenever we reach a $\aplus$-gate, we pay constant delay and memory and reach
      a gate which is an \emph{exit}, i.e., not a $\ctimes$-gate or
      $\aplus$-gate;
    \item Whenever we reach a $\atimes$-gate, we pay constant delay and memory to reach
      two other gates, and we will enumerate an assignment which is a disjoint
      union of the assignments enumerated at each gate, none of which is the
      empty assignment;
    \item Whenever we reach an $\avar$-gate, we pay constant delay and memory to
      enumerate a singleton.
  \end{itemize}
  Hence, when enumerating an assignment $a$, we reach exactly $\card{a}$
  $\avar$-gates, and at most $\card{a}-1$ $\atimes$-gates, so we reach at most
  $\card{a} + (\card{a}-1)$ $\ctimes$-gates and the same number of
  $\aplus$-gates, hence the total delay and memory is linear in the output assignment.
  This concludes the description of the enumeration scheme, which has delay and
  memory linear in each assignment.

  \medskip

  We must now explain how updates are handled.
  Let $g$ be the Boolean variable whose value should be toggled in~$\nu$. We
  modify~$\nu$ to~$\nu'$, use Lemma~\ref{lem:omega} to update $\omega_\nu$
  to~$\omega_{\nu'}$ in time $O(\Delta(C))$, and use Lemma~\ref{lem:switchboard}
  to update $\beta_{\nu}$ to $\beta_{\nu'}$ in same time bound. Further, 
  looking at the proof of Lemma~\ref{lem:switchboard},
  we know that the set $E'$ of edges $e = (g_1, g_2)$ of~$E$
  such that $\beta_\nu(e) \neq \beta_{\nu'}(e)$ must all be such that their second
  gate $g_2$ is in~$\Delta(C)$  Let us split~$E'$ into the edges $E_+
  \colonequals \{e \in E \mid \beta_{\nu'}(e) = 1 \land \beta_\nu(e) = 0\}$
  that are added in~$\beta_{\nu'}(B)$, and a set of edges~$E_-$ analogously defined
  that are deleted in~$\beta_{\nu'}(B)$; each of these edges satisfies that
  their second gate is in~$\Delta(g)$. 
  We update the indexing
  structure of Section~\ref{sec:enumeration} by first deleting~$E_-$, and then
  adding~$E_+$: the end result is still a forest by
  Claim~\ref{clm:switchboard_forest}, and the intermediate result is also a
  forest because we have performed deletions on a forest. We must now argue why
  each of these operations has the required complexity, i.e., $O(\Delta(C))$.
  To see why, observe that the ancestry of~$E_-$
  in the reachability forest before the deletions consists of~$E_-$ plus edges
  where both endpoints are in~$\Delta(g)$, so the ancestry has size $O(\Delta(C))$.
  Likewise, the ancestry of~$E_+$ in the reachability forest before the
  insertions is a subset of the ancestry before the deletions, and in this case
  again it consists of~$E_+$ plus edges where both endpoints are in~$\Delta(g)$,
  hence again the ancestry has size $O(\Delta(C))$.
  Hence, the result of Theorem~\ref{thm:forest} ensures that the complexity of
  updating the reachability structure is still in~$O(\Delta(C))$. This completes
  the description of updates, and the overall update complexity is indeed
  $O(\Delta(C))$.
\end{proof}

We can now show our main result:

\begin{proof}[Proof of Theorem~\ref{thm:main}]
  Given the tree alphabet $\Gamma$ and MSO query $Q(\mathbf{X})$, we use
  Lemma~\ref{lem:balancing} to compute a tree alphabet $\Gamma' \supseteq
  \Gamma$ and MSO query $Q'(\mathbf{X})$. Now, given an input $\Gamma$-tree~$(T,
  \lambda_0)$, we compute in linear time from~$T$ the $\Gamma'$-tree~$T'$
  described by the lemma statement, and we compute in linear time $\lambda''_0$
  which is the valuation $\lambda''$ in the statement of
  Lemma~\ref{lem:balancing} defined from the initial
  valuation $\lambda_0$ of~$T$.
  Now, we can enumerate $Q'$ on $\lambda''(T')$ instead of~$Q$
  on~$\lambda(T)$, and whenever an update operation changes~$\lambda$, then it
  takes constant time to translate it to an update on~$\lambda''$. Hence, we can
  work only with $Q'$, $T'$, and $\lambda''$, without changing our bounds; and
  we know that $T'$ is balanced, i.e., $h(T') = O(\log(\card{T}))$.

  We now use Theorem~\ref{thm:provenance} to compute a hybrid circuit $C$
  capturing the provenance of~$Q''$ on the unlabeled tree~$T'$. We know that
  $C$ is an upwards-deterministic d-DNNF with constant fan-in, and that its
  dependency size is in $O(\log(\card{T}))$. We do this as part of the
  linear-time preprocessing, computing also an initial Boolean valuation $\nu$
  of~$C_{\bvar}$ from the initial valuation $\lambda''$ of~$T'$.
  The definition of provenance
  circuits then ensures that we can enumerate $Q'(\lambda''(T'))$ simply by
  enumerating $\nu(C)$, and that we can reflect updates of~$\lambda''$ by
  translating them in constant time to an update on~$\nu$.

  We now use Lemma~\ref{lem:homogenize} to make
  the circuit homogenized while
  ensuring that it is still an upwards-deterministic d-DNNF and that it still
  satisfies the bound on fan-in and dependency size: note that this adds a
  secondary output gate.
  We now conclude our proof by appealing to Theorem~\ref{thm:circuits}: we can
  enumerate the assignments of~$C$ with linear-time preprocessing, delay and
  memory linear in each produced assignment, and handle updates to~$\nu$ in time
  linear in the dependency size of~$C$, that is, in $O(\log \card{T})$. This
  result ignores the secondary output added when homogenizing the circuit, so we
  may miss the empty assignment whenever it is captured,
  but we can simply extend Theorem~\ref{thm:circuits} to handle the secondary output
  gate $g_1$ by starting the enumeration with the empty assignment if we have
  $\omega_\nu(g_1) = 1$. This achieves the desired bounds, and concludes the
  proof.
\end{proof}

\end{toappendix}

\section{Applications}
\label{sec:applications}
We have finished the proof of our main result (Theorem~\ref{thm:main}),
and now explain how it applies to query languages motivated by
applications. Specifically, we show how to extend our techniques 
to support \emph{aggregate} queries in arbitrary semirings, following the ideas
of semiring provenance~\cite{green2007provenance} and provenance
circuits~\cite{deutch2014circuits}. We then extend this to \emph{group-by queries}, and
last explain how updates are useful to support \emph{parameterized queries}.
Throughout this section, unlike the rest of the paper, we only study MSO queries
with free first-order variables.

\subparagraph*{Aggregate queries.}
We will describe aggregation operators using a general structure called a
\emph{semiring} (always assumed to be commutative). It consists of a \emph{set}
$K$ (finite or infinite), two binary operations $\splus$ and $\stimes$, and
distinguished
elements $0_K, 1_K \in K$. We require that $(K, \splus)$ and $(K, \stimes)$ are
commutative monoids with neutral elements respectively $0_K$ and~$1_K$; that
$\stimes$ distributes over~$\splus$, and that $0_K$ is absorptive for~$\stimes$,
i.e., $0_K \stimes a = 0_K$ for all $a \in K$. We always assume that
evaluating $\splus$ or $\stimes$ take constant time, and that
elements from~$K$ take constant space. Examples of semirings include the natural
numbers $\NN$ with usual addition and product (assumed to take unit time in the RAM
model); or the \emph{security
semiring}~\cite{foster2008annotated}, the \emph{tropical
semiring}~\cite{deutch2014circuits}, etc. Note that sets
of assignments with union and relational product are also a semiring, but one
that does not satisfy our constant-space assumption.

To define aggregation in a semiring~$K$ on a tree~$T$, we consider a mapping
$\rho:T\to K$ giving a value in~$K$ to each node. We extend $\rho$ to
tuples~${\mathbf{b}}$
of~$T$ by setting $\rho({\mathbf{b}}) \colonequals \bigotimes_{n \in
{\mathbf{b}}} \rho(n)$;
to assignments $A$ on some first-order variable set~$\mathbf{x}$ by setting
$\rho(A) \colonequals \bigotimes_{\langle x_i: n \rangle \in A} \rho(n)$;
and to sets $S$ of
assignments by setting $\rho(S) \colonequals \bigoplus_{A \in S} \rho(A)$. An
\emph{aggregate query} on $\Gamma$-trees consists of a
semiring~$K$ (satisfying our assumptions) and of a MSO query $Q(\mathbf{x})$ on $\Gamma$-trees. Given a
$\Gamma$-tree $T$ and a mapping $\rho:T\to K$,
the \emph{aggregate output} $Q_\rho(T)$ of~$Q$ on~$T$ under~$\rho$ is
$\rho(Q(T))$, where $Q(T)$ is the output of~$Q$ on~$T$ as we studied so far, i.e., the set of
assignments~$A$ such that $T \models Q(A)$.
Aggregate MSO queries on trees were already studied,
e.g., by Arnborg and Lagergren~\cite{arnborg1991easy}, but our techniques allow
us to handle updates:

\begin{theoremrep}
  \label{thm:aggregates}
  For any aggregate query $Q(\mathbf{x})$ on $\Gamma$-trees
  with semiring~$K$, given a
  $\Gamma$-tree~$T$ and mapping $\rho:T\to K$, we can compute
  $Q_\rho(T)$ in time $O(\card{T})$, and recompute it in time
  $O(\log \card{T})$ after any update that relabels a node
  of~$T$ or that changes $\rho(n)$ for a node~$n$ of~$T$.
\end{theoremrep}

\begin{proofsketch}
  We adapt hybrid circuits by replacing set-valued gates by $K$-valued gates.
  Now, the set $\gsetv{\nu}{g}$ captured by a gate~$g$ under a Boolean valuation $\nu$
  is an element of~$K$, so we can simplify
  our linear-time preprocessing by making $\omega_\nu$ compute exactly
  $\gsetv{\nu}{g}$ for each gate~$g$. We can then handle updates to~$\nu$ as before,
  and handle updates to~$\rho$ by recomputing~$\omega_\nu$ bottom-up. All of
  this still relies on the balancing lemma
  (Lemma~\ref{lem:balancing}).
\end{proofsketch}

\begin{proof}
  As explained in the sketch, the first step is to show the analogue of
  Theorem~\ref{thm:circuits} where we want to compute $\rho(\gsetv{\nu}{C})$
  instead of enumerating $\gsetv{\nu}{C}$, and where updates can additionally
  change~$\rho$.
  In this variant, we do not apply the homogenization result
  (Lemma~\ref{lem:homogenize}), so we work with a
  hybrid circuit that directly captures the set of assignments of which we want to
  compute the $\rho$-image (i.e., the empty assignment is captured directly,
  without the need for a secondary output).
  We can then perform the initial computation with a much simpler
  variant of the preprocessing scheme, namely, we compute a function
  $\omega_\nu'$ that maps each Boolean gate of~$C$ to its Boolean value
  $\bval{\nu}{g}$ under~$\nu$, and maps each set-valued gate of~$C$ to the value
  $\rho(\gsetv{\nu}{C}) \in K$. We compute $\omega_\nu'$ bottom-up using the
  analogue of Lemma~\ref{lem:omega} (note that this did not depend on
  homogenization of the input circuit), changing the computation on set-valued
  gates as follows:

  \begin{itemize}
    \item For a set-valued variable gate $g$, we set $\omega_\nu'(g)
      \colonequals \rho(g)$;
    \item For a $\atimes$-gate $g$ with no inputs, we set $\omega_\nu'(g)
      \colonequals 1_K$;
    \item For a $\atimes$-gate $g$ with two inputs $g'$ and $g''$, we set
      $\omega_\nu'(g) \colonequals \omega_\nu'(g') \otimes \omega_\nu'(g'')$;
    \item For a $\ctimes$-gate $g$, letting $g'$ be its Boolean input and $g''$
      be its set-valued input, we set $\omega_\nu'(g) \colonequals
      \omega_\nu'(g'')$ if~$\omega_\nu(g') = 1$ and $\omega_\nu'(g) \colonequals
      0_K$ otherwise;
    \item For a $\aplus$-gate $g$, letting $g_1, \ldots, g_n$ be its inputs,
      we set $\omega_\nu'(g) \colonequals \bigoplus_i \omega_\nu'(g')$.
  \end{itemize}

  It is clear by induction that this computes the right value, and the
  computation takes linear time overall because semiring operations take
  constant time by our assumptions.

  Now, whenever an update is performed on a variable gate $g$ (either a set-valued gate,
  for updates to~$\rho$, or a Boolean gate, for updates to~$\nu$), it is clear
  (like in Lemma~\ref{lem:omega}) that the only gates whose $\omega_\nu'$-value
  may change are those of $\Delta(g)$, so we can simply recompute $\omega_\nu'$
  on~$\Delta(g)$ in time $O(\card{\Delta(g)})$.

  We can then conclude the proof using this variant of
  Theorem~\ref{thm:circuits} like we proved Theorem~\ref{thm:main} from
  Theorem~\ref{thm:circuits}, except that we do not apply
  Lemma~\ref{lem:homogenize}. In particular, we make sure to apply
  Lemma~\ref{lem:balancing} before invoking the enumeration result on circuits,
  to ensure that the height of the input tree, hence the dependency size of the
  circuit and the time bound on updates, are in $O(\log \card{T})$: we can do
  this because Lemma~\ref{lem:balancing} preserves exactly the set of
  assignments, so it also preserves the $\rho$-image of this set.
\end{proof}

One important application of this result is \emph{maintaining} the number of query
answers under updates, a question left open by~\cite{losemann2014mso}. We answer
the question for relabeling updates (and in the set semantics), using the
semiring $\NN$ and mapping each node to~$1$ with~$\rho$:

\begin{corollaryrep}
  \label{cor:counting}
  For any MSO query $Q(\mathbf{x})$ on $\Gamma$-trees, given a $\Gamma$-tree $T$, we can compute
  the number $\card{Q(T)}$ of answers of~$Q$ on~$T$ in time $O(\card{T})$, and
  we can update it in time $O(\log \card{T})$ after a relabeling of~$T$.
\end{corollaryrep}

\begin{proof}
  We apply Theorem~\ref{thm:aggregates} using the semiring $\NN$ with usual
  addition and product (assumed to take unit time in the RAM model) and the
  mapping $\rho$ that maps each node to~$1$. This ensures that, for each
  assignment $A$ (including the empty assignment), we have $\rho(A) = 1$; hence,
  for each set $S$ of assignments, we have $\rho(S) = \card{S}$, the number of
  assignments in the set. Thus, Theorem~\ref{thm:aggregates} implies the desired
  result.
\end{proof}

However, we can also use Theorem~\ref{thm:aggregates} for more complex
aggregation semirings:

\begin{example}
  \label{exa:average}
  Let $\Gamma = \{A, B\}$,
  let $Q(x)$ be a MSO query with one variable that selects some tree nodes
  (e.g., select the $B$-labeled nodes which are descendants of some $A$-labeled
  node), let~$(T, \lambda)$ be a $\Gamma$-tree, and
  let $\chi$ be a function that maps each node of~$T$ to an element of the set $\mathbb{D}$ of floating-point numbers (with fixed
  precision). We can compute in linear time the \emph{average} of~$\chi(n)$ for
  the nodes~$n$ such that $T \models Q(n)$, and update it in logarithmic time
  when relabeling a node of~$T$ or changing a value of~$\chi$. This follows from
  Theorem~\ref{thm:aggregates}: we use the semiring of pairs in~$\NN\times \mathbb{D}$ 
  and the mapping $\rho:n \mapsto (1, \chi(n))$ to compute and
  maintain the number of selected nodes and the sum of their $\chi$-images,
  from which we can deduce the average in constant time.
\end{example}

\subparagraph*{Group-by.}
We have adapted our techniques to show results for aggregate queries under
updates. However, supporting updates is also useful for \emph{group-by queries}.
A \emph{group-by query} consists of a MSO query $Q(\mathbf{x}, \mathbf{y})$ on
$\Gamma$-trees with two tuples of first-order variables, and of a semiring $K$.
A \emph{group} on a $\Gamma$-tree $T$ is a set of tuples
$\mathcal{G}(\mathbf{b}) \colonequals \{(\mathbf{b}, {\mathbf{c}}) \mid T \models
Q(\mathbf{b}, {\mathbf{c}})\}$ for some
tuple~$\mathbf{b}$ of nodes of~$T$.
The \emph{output} $Q_\rho(T)$ of~$Q$ on~$T$ under a mapping $\rho:T \to K$
contains one pair $(\mathbf{b}, \rho(\mathcal{G}(\mathbf{b})))$ for each
tuple~$\mathbf{b}$ such that
$\mathcal{G}(\mathbf{b})$ is non-empty.

\begin{example}
  Consider a MSO query $Q(x, y)$ and the semiring $\NN$. The output
  of~$Q$ on a $\Gamma$-tree~$T$ under a mapping~$\rho$ contains one pair per
  $n \in T$, annotated with the sum of $\rho(n')$ for $n'\in T$ such that
  $T \models Q(n, n')$, where we exclude the nodes~$n$ for which the sum is
  empty.
\end{example}

\begin{theoremrep}
  \label{thm:groupby}
  For any group-by query $Q(\mathbf{x}, \mathbf{y})$ 
  and semiring~$K$, given a
  $\Gamma$-tree $T$ and $\rho:T\to K$, we can enumerate 
  $Q_\rho(T)$ with linear-time preprocessing and delay in $O(\log \card{T})$
\end{theoremrep}

\begin{proofsketch}
  We use two enumeration structures. First, we prepare the structure of
  Theorem~\ref{thm:aggregates} for $Q(\mathbf{x}, \mathbf{y})$ but writing the
  valuation of~$\mathbf{x}$ as part of the tree label. Second, 
  we enumerate the non-empty groups with constant delay using Theorem~\ref{thm:main}
  on~$\exists \mathbf{y} ~ Q(\mathbf{x}, \mathbf{y})$. For each tuple
  $\mathbf{b}$ in the output of the second structure, 
  letting $\mathcal{G}(\mathbf{b})$ be the corresponding group,
  we update the
  first structure to compute $\rho(\mathcal{G}(\mathbf{b}))$ in time $O(\log \card{T})$.
\end{proofsketch}

\begin{proof}
  Fix the group-by query $Q(\mathbf{x}, \mathbf{y})$  and the tree alphabet
  $\Gamma$. Let $\Gamma_{\mathbf{x}}$ be the alphabet where we add one label
  $l_{x_i}$ for each $x_i \in \mathbf{x}$. Let $Q'(\mathbf{y})$ be the query
  obtained from~$Q$ by reading the valuation of~$\mathbf{x}$ on the tree using
  the labels $l_{x_i}$, i.e., we add a conjunct asserting that, for each~$i$,
  there is exactly one tree node carrying label $l_{x_i}$, and we quantify
  over~$\mathbf{x}$ so that $x_i$ is interpreted as this one node. It is clear
  that for any unlabeled tree~$T$ and labeling $\lambda:T\to \Gamma$, for each
  tuple~$\mathbf{b}$ of nodes of~$T$,
  letting $\lambda_{\mathbf{b}}:T \to \Gamma_{\mathbf{x}}$
  be the valuation of~$T$ defined by $\lambda_{\mathbf{b}}:n \mapsto \lambda(n) \cup
  \{l_{x_i} \mid b_i = n\}$, we have that $\rho(Q'(\lambda_{\mathbf{b}}(T)))$ is equal to
  $\rho(\mathcal{G}({\mathbf{b}}))$ for the group $\mathcal{G}({\mathbf{b}})$
  associated to~${\mathbf{b}}$ 
  on~$\lambda(T)$. Hence, let us
  apply Theorem~\ref{thm:aggregates} to the query~$Q'$, the semiring~$K$, and
  the mapping~$\rho$, on the tree
  $\lambda_{{\mathbf{b}}_0}(T)$ for some arbitrary choice of~${\mathbf{b}}_0$. We do this as part of our
  linear-time preprocessing, and this describes the first enumeration structure.

  We now describe the second enumeration structure. We consider the query
  $Q''(\mathbf{x}) \colonequals \exists \mathbf{y} ~ Q(\mathbf{x}, \mathbf{y})$.
  It is clear that, for any $\Gamma$-tree $T'$, the output $Q''(T')$ of~$Q''$
  on~$T'$ consists of the tuples ${\mathbf{b}}$ such that the group
  $\mathcal{G}({\mathbf{b}})$ of~$Q$ on~$T'$ is
  non-empty. Hence, we apply
  Theorem~\ref{thm:main} to this query, as part of our linear-time
  preprocessing.

  We now enumerate the non-empty groups as follows. We first enumerate the output of
  $Q''(\lambda(T))$ in constant-delay using the second enumeration structure.
  Each produced tuple ${\mathbf{b}}$ corresponds to a non-empty group
  $\mathcal{G}({\mathbf{b}})$.
  We now modify the labeling function used in
  the first enumeration structure to~$\lambda_{\mathbf{b}}$. To do so, we must change at
  most $2m$ labels, where $m$ is the arity of~$\mathbf{x}$; as $m$ is a
  constant, this is a constant number of updates, so the complexity of doing
  this update on the first structure is in $O(\log \card{T})$. Now, the first
  structure gives us the aggregation value $\rho(\mathcal{G}({\mathbf{b}}))$, 
  and we can produce the pair $({\mathbf{b}}, \rho(\mathcal{G}({\mathbf{b}})))$
  with delay $O(\log \card{T})$. This concludes
  the description of the enumeration phase, and concludes the proof.
\end{proof}

\subparagraph*{Parameterized queries.}
We conclude by presenting another kind of practical queries that we can support
thanks to updates. A \emph{parameterized} MSO query $Q(\mathbf{x}, \mathbf{y})$
on $\Gamma$-trees has two kinds of first-order variables, like group-by:
we call $\mathbf{x}$ the \emph{parameters}. The idea is that, given a
$\Gamma$-tree $T$, the user chooses a tuple ${\mathbf{b}}$ to instantiate the
parameters~$\mathbf{x}$, and we must enumerate efficiently the results of
$Q({\mathbf{b}},
\mathbf{y})$; however the user can change their mind and modify ${\mathbf{b}}$ to change
the value of the parameters. We know by Theorem~\ref{thm:main} that we can
support these queries efficiently: after a
linear-time preprocessing of~$T$, we can enumerate the results
of~$Q({\mathbf{b}}, \mathbf{y})$ with constant delay; and we can react to
changes to~${\mathbf{b}}$ in
time $O(\log \card{T})$ by performing an update on the enumeration structure.

\section{Conclusion}
\label{sec:conclusion}
We have studied MSO queries on trees under \emph{relabeling}
updates, and shown how to enumerate their answers with linear-time
preprocessing, delay and memory linear in each valuation, and update time
logarithmic in the input tree. We have shown this by extending our circuit-based
approach~\cite{amarilli2017circuit} to hybrid circuits, and we have deduced
consequences for practical query languages, in particular for efficient
aggregation. Our results have another technical property that we have not
presented in the main text: like those of~\cite{losemann2014mso}, they are also
tractable in the size of the query when representing it as a deterministic automaton.

The main direction for future work would be to extend our result to support
insertions and deletions of leaves, like~\cite{losemann2014mso}, hopefully 
preserving our improved bounds:
while deletions can be emulated with relabelings, insertions are trickier.
Such a result was very recently shown in~\cite{niewerth2018enumeration} for the case of
words rather than trees.
We believe that many of our constructions on trees should
adapt to insertions and deletions.
The main challenge is to extend Lemma~\ref{lem:balancing}, which
we believe to be an interesting question in its own right: the technique
of~\cite{balmin2004incremental} may be applicable here, although it would lead to an
$O(\log^2 n)$ update time.

\setcounter{figure}{4}
\renewcommand{\figurename}{}
\renewcommand\thefigure{(\alph{figure})}
\captionsetup{labelformat=simple}

\bibliography{main}
\end{document}